\documentclass{amsart}[11pt]
\usepackage{amsmath, amsfonts, amsthm, amssymb,mathtools}
\usepackage{tabularx,ragged2e,booktabs,caption}
\usepackage{subfigure}
\usepackage{stmaryrd}
\usepackage{verbatim}
\usepackage{hyperref}
\usepackage{color}
\usepackage{ulem}
\usepackage[dvips]{graphicx}
\usepackage{enumerate}
\usepackage{etaremune}

\numberwithin{equation}{section}
\newtheorem{theorem}{Theorem}[section]
\newtheorem{corollary}[theorem]{Corollary}
\newtheorem{lemma}[theorem]{Lemma}

\theoremstyle{definition}

\newtheorem{remark}[theorem]{Remark}

\newcommand{\EE}{\mathbb{E}}
\newcommand{\RR}{\mathbb{R}}

\newcommand{\PP}{\mathbb{P}}
\newcommand{\D}{\mathrm{d}}

\newcommand{\E}{\mathrm{e}}
\newcommand{\Ee}{\mathcal{E}}

\newcommand{\JJ}{\mathrm{J}}

\newcommand{\Nn}{\mathcal{N}}
\newcommand{\Ff}{\mathcal{F}}
\newcommand{\Vv}{\mathcal{V}}

\newcommand{\II}{\mathrm{I}}
\newcommand{\mm}{\mathrm{m}}
\newcommand{\BS}{\mathrm{BS}}

\newcommand{\CEV}{\mathrm{CEV}}
\newcommand{\ttau}{\widetilde{\tau}}

\newcommand{\yyp}{\overline{y}_{p}}

\begin{document}

\title{Black-Scholes in a CEV random environment}
\author{Antoine Jacquier and Patrick Roome}
\address{Department of Mathematics, Imperial College London}
\email{a.jacquier@imperial.ac.uk, p.roome11@imperial.ac.uk}
\date{\today}
\thanks{AJ acknowledges financial support from the EPSRC First Grant EP/M008436/1.}
\maketitle
\begin{abstract}
Classical (It\^o diffusions) stochastic volatility models are not able to capture the steepness of small-maturity implied volatility smiles.
Jumps, in particular exponential L\'evy and affine models, which exhibit small-maturity exploding smiles,
have historically been proposed to remedy this (see~\cite{Tank} for an overview), 
and more recently rough volatility models~\cite{AlosLeon, Fukasawa}.
We suggest here a different route, randomising the Black-Scholes variance by a CEV-generated  distribution,
which allows us to modulate the rate of explosion (through the CEV exponent)
of the implied volatility for small maturities.
The range of rates includes behaviours similar to exponential L\'evy models and fractional stochastic volatility models.
\end{abstract}

%%%%%%%%%%%%%%%%%%%%%%%%%%%%%%%%%%%%%%%%
%%%%%%%%%%%%%%%%%%%%%%%%%%%%%%%%%%%%%%%%
\section{Introduction}\label{sec:introcevvar}
We propose a simple model with continuous paths for stock prices that allows for 
small-maturity explosion of the implied volatility smile.
It is indeed a well-documented fact on Equity markets (see for instance~\cite[Chapter~5]{GatheralBook})
that standard (It\^o) stochastic models with continuous paths are not able to 
capture the observed steepness of the left wing of the smile when the maturity becomes small.
To remedy this, several authors have suggested the addition of jumps, either in the form of an independent L\'evy process or within the more general framework of affine diffusions.
Jumps (in the stock price dynamics) imply an explosive behaviour for the small-maturity smile and are better able to capture the observed steepness of the small-maturity implied volatility smile.
In particular, Tankov~\cite{Tank} showed that, for exponential L\'evy models with L\'evy measure supported
on the whole real line, 
the squared implied volatility smile explodes as $\sigma_{\tau}^2(k) \sim -k^2/(2\tau \log \tau) $,
as the maturity~$\tau$ tends to zero, where $k$ represents the log-moneyness. 
Such a small-maturity behaviour of the smile is not only captured by jump-based models,
but rough volatility (non-Markovian) models, where the stochastic volatility component
is driven by a fractional Brownian motion, are in fact also able to reflect this property of the data.
In a series of papers several authors~\cite{AlosLeon, BFGHS, FordeZhang, Fukasawa,  GJRIC14, Guli17, JPS17} have indeed proved that,  when the Hurst index of the fractional Brownian motion lies within $(0,1/2)$, 
then the implied volatility explodes at a rate of 
${\tau}^{H-1/2}$ as the maturity~$\tau$ tends to zero.

In this paper we propose an alternative framework:
we suppose that the stock price follows a standard Black-Scholes model;
however the instantaneous variance, instead of being constant, is sampled from a continuous distribution.
We first derive some general properties, interesting from a financial modelling point of view,
and devote a particular attention to a particular case of it, 
where the variance is generated from independent CEV dynamics.
Assume that interest rates and dividends are null, and 
let $S$ denote the stock price process starting at $S_0 = 1$, 
the solution to the stochastic differential equation
$\D S_\tau = S_\tau \sqrt{\Vv}\D W_\tau$, for $\tau\geq 0$, 
where~$W$ is a standard Brownian motion.
Here, $\Vv$ is a random variable, which we assume to be distributed as $\Vv\sim Y_{t}$,
for some~$t>0$, 
where $Y$ is the unique strong solution of the CEV dynamics
$\D Y_u = \xi Y_u^{p}\D B_u$, $Y_0>0$
where $p\in\RR$, $\xi>0$ and $B$ is an independent Brownian motion 
(see Section~\ref{sec:Model} for precise statements).
The main result of this paper (Theorem~\ref{theorem:genfwdsm}) is that 
the implied volatility generated from this model exhibits the following behaviour as the maturity~$\tau$ tends to zero:
\begin{align}\label{eq:genblowup}
\sigma_{\tau}^2(k)  & \sim 
\left\{
\begin{array}{ll}
\displaystyle \frac{2(1-p)}{3-2p}\left(\frac{k^2\xi^2 (1-p)t}{2\tau}\right)^{1/(3-2p)}, & \text{if }p<1,\\
\displaystyle \frac{k^2 \xi^2 t}{\tau(\log \tau)^2 }, & \text{if }p=1, \\
\displaystyle \frac{k^2}{2(2p-1) \tau |\log\tau|}, & \text{if }p>1,
\end{array}
\right.
\end{align}
for all $k\ne 0$.
Sampling the initial variance from the CEV process at time $t$ induces different term structures
for small-maturity spot smiles, thereby providing flexibility to match steep small-maturity smiles.
For $p>1$, the explosion rate is the same as exponential L\'evy models,
and the case $p\leq 1/2$ mimics the explosion rate of fractional stochastic volatility models.
The CEV exponent~$p$ therefore allows the user to modulate the short-maturity steepness of the smile.

We are not claiming here that this model should come as a replacement of fractional stochastic volatility models
or exponential L\'evy models, notably because its dynamic structure looks too simple at first sight.
However, we believe it can act as an efficient building block for more involved models, 
in particular for stochastic volatility models with initial random distribution for the instantaneous 
variance.
While we leave these extensions for future research, we shall highlight how our model
comes naturally into play when pricing forward-start options in stochastic volatility models.
In~\cite{JR2013} the authors proved that the small-maturity forward implied volatility smile explodes in the Heston model when the remaining maturity (after the forward-start date) becomes small.
This explosion rate corresponds precisely to the case $p=1/2$ in~\eqref{eq:genblowup}.
This in particular shows that the key quantity determining the explosion rate 
is the (right tail of the) variance distribution at the forward-start date (here corresponding to~$t$).

The paper is structured as follows:
in Sections~\ref{sec:Model} and~\ref{sec:cevrandom} we introduce our model and relate it to other existing approaches.
In Section~\ref{sec:MGFApproach} we use the moment generating function to derive extreme strike asymptotics 
(for some special cases) and show why this approach is not readily applicable for small and large-maturity asymptotics.
Sections~\ref{sec:smalltimegeneralcev} and~\ref{sec:LargeTime} detail the main results, 
namely the small and large-maturity asymptotics of option prices and the corresponding implied volatility.
Section~\ref{sec:numericsgenfwdsmile} provides numerical examples,
and Section~\ref{sec:forwardcev} describes the relationship between our model 
and the pricing of forward-start options in stochastic volatility models.
Finally, the proofs of the main results are gathered in Section~\ref{sec:ProofsCEVgeneral}.

\textbf{Notations}:
Throughout the paper, the $\sim$ symbol means asymptotic equivalence, namely, the ratio of the left-hand side to the right-hand side tends to one.

%%%%%%%%%%%%%%%%%%%%%%%%%%%%%%%%%%%%%%%%%%%%%%%%
%%%%%%%%%%%%%%%%%%%%%%%%%%%%%%%%%%%%%%%%%%%%%%%%
\section{Model and main results}

%%%%%%%%%%%%%%%%%%%%%%%%%%%%%%%%%%%%%%%%%%%%%%%%
%%%%%%%%%%%%%%%%%%%%%%%%%%%%%%%%%%%%%%%%%%%%%%%%
\subsection{Model description}\label{sec:Model}
We consider a filtered probability space $(\Omega, \Ff, (\Ff_s)_{s\geq 0}, \PP)$
supporting a standard Brownian motion, and let $(Z_{s})_{s\geq 0}$ 
denote the solution to the following stochastic differential equation:
\begin{equation}\label{eq:Model}
\D Z_s  = -\frac{1}{2}\Vv\D s+ \sqrt{\Vv}\D W_s, 
\qquad Z_0 = 0,
\end{equation}
where $\Vv$ is some random variable, independent of the Brownian motion~$W$, 
and in particular of the Brownian filtration at time zero 
(see~\cite[Remarks~2.2 and~2.3]{JS16} for details about this).
The process $(Z_s)_{s\geq 0}$, in finance, corresponds to the logarithm of 
the underlying stock price, and the coefficient $-1/2$ ensures (up to integrability properties of~$\E^{\Vv}$)
that
$(\E^{Z_s})_{s\geq 0}$ is a true $(\Ff_s)_{s\geq 0}$-martingale.
In the case where~$\Vv$ is a discrete random variable, this model reduces to the mixture of distributions,
analysed, in the Gaussian case by Brigo and Mercurio~\cite{Brigo, BrigoMercurio}.
In a stochastic volatility model where the instantaneous variance process~$(V_t)_{t\geq 0}$ 
is uncorrelated with the asset price process, 
the mixing result by Romano and Touzi~\cite{RT97} implies that the price of a European option with maturity~$\tau$
is the same as the one evaluated from the SDE~\eqref{eq:Model} with
$\Vv = \tau^{-1}\int_0^{\tau}V_s \D s$.
As $\tau$ tends to zero, the distribution of $\Vv$ approaches a Dirac Delta
centred at the initial variance~$V_0$.
Asymptotics of the implied volatility are well known and weaknesses of classical stochastic volatility models 
are well documented~\cite{GatheralBook}.
Although such models fit into the framework of~\eqref{eq:Model}, we will not consider them further in this paper.
Define pathwise the process $M$ by $M_s := -\frac{1}{2}s + W_s$ and let $(\mathcal{T}_s)_{s\geq0}$ be given 
by $\mathcal{T}_s := s\Vv$.
Then $\mathcal{T}$ is an independent increasing time-change process and $Z=M_{\mathcal{T}}$ in distribution.
In this way our model can be thought of as a random time change.
Let now~$N$ be a L\'evy process such that $(\E^{N_s})_{s\geq0}$ is a $(\Ff_s)_{s\geq 0}$-adapted martingale;
define $\Vv := \tau^{-1}\int_0^{\tau}V_s \D s$ where~$V$ is a positive and independent process, 
then $(\E^{N_{\mathcal{T}_s}})_{s\geq0}$ is a classical time-changed exponential L\'evy process,
and pricing vanilla options is standard~\cite[Section 15.5]{CT07}.
However, as the maturity~$\tau$ tends to zero, $\Vv$ converges in distribution to a Dirac Delta,
in which case asymptotics are well known~\cite{Tank}.

The model~\eqref{eq:Model} is also related to the Uncertain Volatility Model of Avellaneda and Par\'as~\cite{Paras}
(see also~\cite{DenisMartini, UVMMartini, Lyons}), in which the Black-Scholes volatility
 is allowed to evolve randomly within two bounds.
In this framework, sub-and super-hedging strategies (corresponding to best and worst case scenarios) 
are usually derived via the Black-Scholes-Barenblatt equation, 
and Fouque and Ren~\cite{FouqueRen} recently provided approximation results 
when the two bounds become close to each other.
One can also, at least formally, look at~\eqref{eq:Model} 
from the perspective of fractional stochastic volatility models, 
first proposed by Comte et al. in~\cite{CR98}, and later developed and revived 
in~\cite{BFG15, CCR12, BLP15, BFGMS17, BLP16, ER16, HJL, Fukasawa, GJR14, GJRIC14, HJMDonsker, JMM17, MP17}.
In these models, standard stochastic volatility models are generalised by replacing the Brownian motion
driving the instantaneous volatility by a fractional Brownian motion.
This preserves the martingale property of the stock price process, and allows, in the case of short memory
(Hurst parameter~$H$ between $0$ and $1/2$) for short-maturity steep skew of the implied volatility smile.
However, the Mandelbrot-van Ness representation~\cite{Mandelbrot} of the fractional Brownian motion reads
$$
W_t^{H} := \int_{0}^{t}\frac{\D W_s}{(t-s)^{\gamma}} + \int_{-\infty}^{0}\left(\frac{1}{(t-s)^\gamma} - \frac{1}{(-s)^\gamma}\right)\D W_s,
$$
for all $t\geq 0$, where $\gamma:=1/2-H$.
This representation in particular indicates that, at time zero, the instantaneous variance, 
being driven by a fractional Brownian motion, incorporates some randomness (through the second integral).
Finally, we agree that, at first sight, randomising the variance may sound unconventional.
As mentioned in the introduction, we see this model as a building block for more involved models, 
in particular stochastic volatility with random initial variance, the full study of which is the purpose
of ongoing research.
After all, market data only provides us with an initial value of the stock price, 
and the initial level of the variance is unknown, usually left as a parameter to calibrate.
In this sense, it becomes fairly natural to leave the latter random.

%%%%%%%%%%%%%%%%%%%%%%%%%%%%%%%%%%%%%%%%%%%%%%%%
%%%%%%%%%%%%%%%%%%%%%%%%%%%%%%%%%%%%%%%%%%%%%%%%
\subsubsection{Moment generating function}\label{sec:LeeWings}
In~\cite{FJ10, FJ11, JKRM2013}, the authors used the theory of large deviations, 
and in particular the G\"artner-Ellis theorem, to prove small-and large-maturity behaviours of the implied volatility in the Heston model and more generally (in~\cite{JKRM2013}) for affine stochastic volatility models.
This approach relies solely on the knowledge of the cumulant generating function of the underlying stock
price, and its rescaled limiting behaviour.
For any $\tau\geq 0$, let $\Lambda^Z(u,\tau):=\log\EE(\E^{uZ_\tau})$ denote the cumulant generating function of $Z_\tau$, defined on the effective domain $\mathcal{D}^Z_\tau:=\{u\in\RR: |\Lambda^Z(u,\tau)|<\infty\}$;
similarly denote $\Lambda^{\Vv}(u)\equiv \log\EE(\E^{u\Vv})$, whenever it is well defined.
A direct application of the tower property for expectations yields
\begin{equation}\label{eq:LambdaXV}
\Lambda^Z(u,\tau) = \Lambda^{\Vv}\left(\frac{u(u-1)\tau}{2}\right), \qquad
\text{for all }u \in \mathcal{D}^Z_\tau.
\end{equation}
Unfortunately, the cumulant generating function of~$\Vv$ is not available in closed-form in general.
In Section~\ref{sec:MGFApproach} below, we shall see some examples where such a closed-form solution is available,
and where direct computations are therefore possible.
We note in passing that this simple representation allows, at least in principle, 
for straightforward (numerical) computations of the slopes of the wings of the implied volatility 
using Roger Lee's Moment Formula~\cite{Lee} (see also Section~\ref{sec:LeeWingsMGF}).
The latter are indeed given directly by the boundaries (in $\RR$) of the effective domain of $\Lambda^\Vv$.
Note further that the model~\eqref{eq:Model} could be seen as a time-changed Brownian motion (with drift); 
the representation~\eqref{eq:LambdaXV} clearly rules out the case where~$Z$ is a simple exponential L\'evy process (in which case $\Lambda^Z(u,\tau)$ would be linear in~$\tau$). 
In view of Roger Lee's formula, this also implies that, contrary to the L\'evy case, 
the slopes of the implied volatility wings are not constant over time in our model.

%%%%%%%%%%%%%%%%%%%%%%%%%%%%%%%%%%%%%%%%%%%%%%%%
%%%%%%%%%%%%%%%%%%%%%%%%%%%%%%%%%%%%%%%%%%%%%%%%
\subsection{CEV randomisation}\label{sec:cevrandom}
As mentioned above, this paper is a first step towards the introduction of `random environment' into
the realm of option pricing, and we believe that, seeing it `at work' through a specific, yet non-trivial, example, will speak for its potential prowess. 
We assume from now on that $\Vv$ corresponds to the distribution
of the random variable generated, at some time~$t$, by the solution to the CEV 
stochastic differential equation
$\D Y_u = \xi Y_u^{p}\D B_u$, $Y_0=y_0>0$
where $p\in\RR$, $\xi>0$ and $B$ is a standard Brownian motion, independent of~$W$.
The CEV process~\cite{BL12, Jeanblanc} is the unique strong solution to
this stochastic differential equation, 
up to the stopping time $\tau^Y_0:=\inf_{u>0}\{Y_u=0\}$.
The behaviour of the process after~$\tau^Y_0$ depends on the value of~$p$, and shall be discussed below.
We let 
$\Gamma(n;x):=\Gamma(n)^{-1}\int_0^x t^{n-1}\E^{-t}\D t$ 
denote the normalised lower incomplete Gamma function, 
and $\mm_t := \PP(Y_t=0) = \PP(\Vv = 0)$ represent the mass at the origin.
Define the constants
\begin{equation}\label{eq:constantsCEV}
\eta := \frac{1}{2(p-1)}, \qquad \qquad 
\mu  :=\log(y_0) - \frac{\xi^2 t}{2}.
\end{equation}
Straightforward computations show that, whenever the origin is an absorbing boundary,
the density $\zeta_p(y) \equiv \mathbb{P}(Y_t\in\D y)/\D y$ is norm decreasing 
and
\begin{equation}\label{eq:mass0}
\mm_t=1-\Gamma\left(-\eta;\frac{y_0^{2(1-p)}}{2\xi^2(1-p)^2 t}\right) > 0;
\end{equation}
otherwise $\mm_t=0$ and the density~$\zeta_p$ is norm preserving.
When $p\in[1/2,1)$, the origin is naturally absorbing.
When $p\geq 1$, the process~$Y$ never hits zero $\PP$-almost surely. 

Finally, when $p<1/2$, the origin is an attainable boundary, 
and can be chosen to be either absorbing or reflecting.
Absorption is compulsory if~$Y$ is required to be a martingale~\cite[Chapter III, Lemma 3.6]{Jacod}.
Here it is only used as a building block for the instantaneous variance, 
and such a requirement is therefore not needed, so that both cases (absorption and reflection)
will be treated.
Introduce the function $\varphi_{\eta}: (0,\infty)\to (0,\infty)$ by
$$
\varphi_{\eta}(y)
 := \frac{y_0^{1/2}y^{1/2-2p}}{|1-p|\xi^2 t}
\exp\left(-\frac{y^{2(1-p)} + y_0^{2(1-p)}}{2\xi^2 t (1-p)^2}\right)
\II_{\eta}\left(\frac{(y_0 y)^{1-p}}{(1-p)^2\xi^2 t}\right),
$$
where $\II_{\eta}$ is the modified Bessel function of the first kind of order $\eta$~\cite[Section 9.6]{Abra}.
The CEV density, $\zeta_p(y) := \mathbb{P}(Y_t\in\D y)/\D y$, then reads
\begin{align}\label{eq:CEVdensity}
\zeta_p(y)= 
\left\{
\begin{array}{ll}
\displaystyle \varphi_{-\eta}(y),
& \text{if }p \in [1/2, 1) \text{ or }p < \frac{1}{2} \text{ with absorption},\\
\displaystyle \varphi_{\eta}(y),
& \text{if }p>1\text{ or }p < \frac{1}{2} \text{ with reflection},\\
\displaystyle \frac{1}{y \xi\sqrt{2\pi t}}
\exp\left(-\frac{(\log(y)-\mu)^2}{2\xi^2 t}\right),
& \text{if }p=1,
\end{array}
\right.
\end{align}
valid for $y\in (0,\infty)$.
When $p\geq 1$, the density $\zeta_p$ converges to zero around the origin, implying that paths are being pushed away from the origin.
On the other hand $\zeta_p$ diverges to infinity at the origin when $p<1/2$, 
so that the paths have a propensity towards the vicinity of the origin.

It is clear from all the quantities above that the precise horizon~$t$ itself is not fundamental,
as it only appears with the multiplicative constant factor~$\xi^2$.
By scaling of the Brownian motion, $t$ can be taken equal to unity, and is therefore rather irrelevant here;
we shall keep it explicit in the notations, however, since it will turn out useful when applying this framework
to forward-start derivatives in Section~\ref{sec:forwardcev}.

%%%%%%%%%%%%%%%%%%%%%%%%%%%%%%%%%%%%%%%%%%%%%%%
%%%%%%%%%%%%%%%%%%%%%%%%%%%%%%%%%%%%%%%%%%%%%%%
\subsection{The moment generating function approach}\label{sec:MGFApproach}
In the literature on implied volatility asymptotics, 
the moment generating function of the stock price has proved to be an extremely useful tool
to obtain sharp estimates.
This is obviously the case for the wings of the smile (small and large strikes) via Roger Lee's formula,
mentioned in Section~\ref{sec:LeeWings},
but also to describe short-and large-maturity asymptotics, as developed for instance in~\cite{JKRM2013}
or~\cite{JR12}, via the use of (a refined version of) the G\"artner-Ellis theorem.
In~\cite{JS16}, the authors used this property to study a generalised version of the Heston model,
where the starting value of the instantaneous volatility is randomised according to some distribution.
It it closed to the present model, yet does not supersede it, and makes full use of the knowledge of the moment generating function of the Heston model.
As shown in Section~\ref{sec:LeeWings}, the moment generating function of a stock price 
satisfying~\eqref{eq:Model} is fully determined by that of the random variable~$\Vv$.
However, even though the density of the latter is known in closed form (Equation~\eqref{eq:CEVdensity}), 
the moment generating function is not so for general values of~$p$.
In the cases $p=0$ (with either reflecting or absorbing boundary) and $p=1/2$,
a closed-form expression is available and direct computations are possible.

%%%%%%%%%%%%%%%%%%%%%%%%%%%%%%%%%%%%%%%%
\subsubsection{Computation of the moment generating function}
Denote by $\Lambda^{\Vv}_{0,r}$, $\Lambda^{\Vv}_{0,a}$ and $\Lambda^{\Vv}_{1/2}$
the moment generating function of the random variable~$\Vv$ when $p=0$ 
(the subscript `r' / `a' denotes the reflecting / absorbing behaviour at the origin) and $p=1/2$.
The following quantities can be computed directly from~\cite[Part~I, Section~6.4]{Jeanblanc}:
\begin{equation}\label{eq:mgfV}
\begin{array}{ll}
\Lambda^{\Vv}_{0,a}(u) & = \displaystyle 
\log\left[
\mm_t + \frac{1}{2}\exp\left(\frac{(u\xi^2 t - 2y_0)u}{2}\right)
\left\{
\E^{2uy_0}\Ee\left(\frac{u\xi^2 t +y_0}{\xi\sqrt{2t}}\right)
+\E^{2uy_0} - 1 
 - \Ee\left(\frac{u\xi^2 t - y_0}{\xi\sqrt{2t}}\right)\right\}\right],\\
\Lambda^{\Vv}_{0,r}(u) & = \displaystyle 
\log\left[\frac{1}{2}\exp\left(\frac{(u\xi^2 t - 2y_0)u}{2}\right)
\left\{
\E^{2uy_0}\Ee\left(\frac{u\xi^2 t +y_0}{\xi\sqrt{2t}}\right)
+\E^{2uy_0}+ 1
 + \Ee\left(\frac{u\xi^2 t - y_0}{\xi\sqrt{2t}}\right)\right\}\right],\\
\Lambda^{\Vv}_{1/2}(u) & = \displaystyle \frac{2y_0 u}{2-u\xi^2 t},
\end{array}
\end{equation}
where $\Ee(z) \equiv \frac{2}{\sqrt{\pi}}\int_{0}^{z}\exp(-x^2)\D x$ is the error function.
Note that when $p=1/2$ and $p=0$ in the absorption case, 
one needs to take into account the mass at zero in~\eqref{eq:mass0} when computing these expectations.

%%%%%%%%%%%%%%%%%%%%%%%%%%%%%%%%%%%%%%%%%%%%%%%
\subsubsection{Roger Lee's wing formula}\label{sec:LeeWingsMGF}
In~\cite{Lee}, Roger Lee provided a precise link between the slope of the total implied variance 
in the wings and the boundaries of the domain of the moment generating function of the stock price.
More precisely, for any $\tau\geq 0$, let $u_{+}(\tau)$ and $u_{-}(\tau)$ be defined as
$$
u_+(\tau) := \sup\{u\geq 1: |\Lambda^Z(u,\tau)|<\infty\}
\qquad\text{and}\qquad
u_-(\tau) := \sup\{u\geq 0: |\Lambda^Z(-u,\tau)|<\infty\}.
$$
The implied volatility $\sigma_\tau(k)$ then satisfies
$$
\limsup_{k\uparrow\infty}\frac{\sigma_\tau(k)^2\tau}{k} = \psi(u_+(\tau)-1)=:\beta_+(\tau)
\qquad\text{and}\qquad
\limsup_{k\downarrow-\infty}\frac{\sigma_\tau(k)^2\tau}{|k|} = \psi(u_-(\tau)))=:\beta_-(\tau),
$$
where the function $\psi$ is defined by 
$\psi(u) = 2-4\left(\sqrt{u(u+1)}-u\right)$.
Combining~\eqref{eq:mgfV} and~\eqref{eq:LambdaXV} yields a closed-form expression
for the moment generating function of the stock price when $p\in\{0,1/2\}$.
It is clear that, when $p=0$, $u_\pm(\tau) = \pm\infty$ for any $\tau\geq 0$,
and hence the slopes of the left and right wings are equal to zero (the total variance flattens 
for small and large strikes).
In the case where $p=1/2$, explosion will occur as soon as
$\left(\frac{1}{2}u(u-1)\tau\xi^2 t - 2\right) = 0$, so that
$$
u_\pm(\tau) = \frac{1}{2} \pm \frac{1}{2}\sqrt{1+\frac{16}{\xi^2 t \tau}},
\qquad\text{and}\qquad
\beta_-(\tau) = \beta_+(\tau) = \frac{2}{\xi\sqrt{t\tau}}\left(\sqrt{\xi^2 t \tau + 16}-4\right),
\qquad\text{for all }\tau >0.
$$
The left and right slopes are the same, but the product~$\xi^2 t$
can be directly calibrated on the observed wings.
Note that the map $\tau\mapsto\beta_{\pm}(\tau)$ is concave and increasing from $0$ to $2$. 
In~\cite{DFJV1, DFJV2}, the authors highlighted some symmetry properties between the small-time behaviour 
of the smile and its tail asymptotics.
We obtain here some interesting asymmetry, in the sense that one can observe the same type of rate of explosion
 (power behaviour, given by~\eqref{eq:genblowup} in the case $p<1$), but different tail behaviour for fixed maturity.
As $\tau$ tends to infinity, $\beta_\pm(\tau)$ converges to $2$,
so that the implied volatility smile does not `flatten out', as is usually the case
for It\^o diffusions or affine stochastic volatility models (see for instance~\cite{JKRM2013}).
In Section~\ref{sec:LargeTime} below, we make this more precise by investigating the large-time behaviour 
of the implied volatility using the density of the CEV-distributed variance.

%%%%%%%%%%%%%%%%%%%%%%%%%%%%%%%%%%%%%%%%%%%%%%%
%%%%%%%%%%%%%%%%%%%%%%%%%%%%%%%%%%%%%%%%%%%%%%%
\subsubsection{Small-time asymptotics}\label{sec:MGFSmallTime}
In order to study the small-maturity behaviour of the implied volatility, 
one could, whenever the moment generating function of the stock price is available in closed form
(e.g. in the case $p\in \{0,1/2\}$), apply the methodology developed in~\cite{FJ10}.
The latter is based on the G\"artner-Ellis theorem, which, essentially, consists of
finding a smooth convex pointwise limit (as $\tau$ tends to zero) of some rescaled version of the cumulant generating function.
In the case where $p=1/2$, it is easy to show that
\begin{equation}\label{eq:SmallTimeMGF12}
\Lambda^Z_0(u) := \lim_{\tau\downarrow 0}\tau^{1/2}\Lambda^Z\left(\frac{u}{\sqrt{\tau}},\tau\right) = 
\left\{
\begin{array}{ll}
0, & \displaystyle \text{if } u \in \left(-\frac{2}{\xi\sqrt{t}}, \frac{2}{\xi\sqrt{t}}\right),\\
+\infty, & \text{otherwise}.
\end{array}
\right.
\end{equation}
The nature of this limiting behaviour falls outside the scope of the G\"artner-Ellis theorem,
which requires $|\Lambda^Z_0(u)|$ to diverge to infinity as~$u$ approaches the boundaries $\pm 2 / (\xi\sqrt{t})$.
It is easy to see that any other rescaling would yield even more degenerate behaviour.
One could adapt the proof of the G\"artner-Ellis theorem, as was done in~\cite{JR2013}
for the small-maturity behaviour of the forward implied volatility smile in the Heston model
(see also~\cite{DMJR15} and references therein for more examples of this kind).
In the case~\eqref{eq:SmallTimeMGF12}, we are exactly as in the framework of~\cite{JR2013},
in which the small-maturity smile (squared) indeed explodes as $\tau^{-1/2}$, 
precisely the same explosion as the one in~\eqref{eq:genblowup}.
Unfortunately, as we mentioned above, the moment generating function of the stock price 
is not available in general, and this approach is hence not amenable here.

%%%%%%%%%%%%%%%%%%%%%%%%%%%%%%%%%%%%%%%%%%%%%%%
%%%%%%%%%%%%%%%%%%%%%%%%%%%%%%%%%%%%%%%%%%%%%%%
\subsubsection{Large-time asymptotics}\label{sec:MGFSmallTime}
The analysis above, based on the moment generating function of the stock price, can be carried over to study the large-time behaviour of the latter.
In the case $p=1/2$, computations are fully explicit, and the following pointwise limit follows from simple
straightforward manipulations:
\begin{equation*}
\lim_{\tau\uparrow\infty}\tau^{-1}\Lambda^Z(u,\tau) = 
\left\{
\begin{array}{ll}
0, & \text{if } u \in [0,1],\\
+\infty, & \text{otherwise}.
\end{array}
\right.
\end{equation*}
The nature of this asymptotic behaviour, again, falls outside the scope of standard large deviations analysis,
and tedious work, in the spirit of~\cite{Bercu, JR2013}, would be needed to pursue this route.

%%%%%%%%%%%%%%%%%%%%%%%%%%%%%%%%%%%%%%%%%%%%%%%
%%%%%%%%%%%%%%%%%%%%%%%%%%%%%%%%%%%%%%%%%%%%%%%
\subsection{Small-time behaviour of option prices and implied volatility}\label{sec:smalltimegeneralcev}
In the Black-Scholes model $\D S_t = \sqrt{w} S_t \D W_t$ starting at $S_0=1$, 
a European call option with strike $\E^{k}$ and maturity $T>0$ is worth
\begin{equation}\label{eq:BSvariance}
\BS(k, w, T)  = \Nn\left(-\frac{k}{\sqrt{w T}}+\frac{\sqrt{wT}}{2}\right)
 - \E^{k}\Nn\left(-\frac{k}{\sqrt{wT}}-\frac{\sqrt{wT}}{2}\right).
\end{equation}
For any $k\in\RR\setminus\{0\}$, $T>0$, and $p>1$, the quantity 
\begin{align}\label{eq:Ip}
\JJ^{p}(k)   & := 
\left\{
\begin{array}{ll}
\displaystyle \int_{0}^{\infty}\BS\left(k, \frac{y}{T}, T\right)y^{-p}\D y, & \text{if }k>0,\\
\displaystyle \int_{0}^{\infty}\Big(\E^k-1+\BS\left(k, \frac{y}{T}, T\right)\Big)y^{-p}\D y, & \text{if }k<0,
\end{array}
\right.
\end{align}
is well defined and independent of~$T$.
Indeed, since the stock price is a martingale starting at one, 
Call options are always bounded above by one,
and hence, for $k>0$,
$\JJ^{p}(k) \leq \int_{0}^{1}\BS(k, y/T, T) y^{-p} \D y + \int_{1}^{\infty}y^{-p} \D y$.
The second integral is finite since $p>1$.
When $k>0$, the asymptotic behaviour 
$$
\BS\left(k,\frac{y}{T}, T\right)\sim \exp\left(-\frac{k^2}{2y} + \frac{k}{2}\right) \frac{y^{3/2}}{k^2\sqrt{2\pi}}
$$
holds as $y$ tends to zero, so that 
$\lim_{y\downarrow 0}\BS(k, y/T, T) y^{-p}=0$, and hence the integral is finite.
A similar analysis holds when $k<0$ and using put-call parity.
Define now the following constants:
\begin{equation}\label{eq:betapalphasmmat}
\beta_p :=   \frac{1}{3-2p}, \qquad
\yyp  : =\left(\frac{k^2\xi^2t(1-p)}{2}\right)^{\beta_p}, \qquad
y^*:= \frac{k^2 \xi^2 t}{2},
\end{equation}
the first two only when $p<1$, and note that $\beta_p \in (0,1)$;
define further the following functions from $(0,\infty)$ to $\RR$:
\begin{equation}\label{eq:f0g0smmat}
\left\{
\begin{array}{rlrl}
f_0(y) & := \displaystyle \frac{k^2}{2y}+\frac{y^{2(1-p)}}{2\xi^2t(1-p)^2},
\qquad \qquad &  f_1(y) & :=\displaystyle \frac{(y y_0)^{(1-p)}}{\xi^2t(1-p)^2}, \\
g_0(y) & := \displaystyle  \frac{k^2}{2y}+\frac{\log(y)}{\xi^2 t},
\qquad \qquad & g_1(y) & := \displaystyle \frac{\log(y)}{\xi^2 t}, 
\end{array}
\right.
\end{equation}
as well as the following ones, parameterised by~$p$: 
\\
\begin{minipage}{\linewidth}
\centering
\begin{tabular}{c| c c c}
& $p<1$ & $p=1$ & $p>1$ \\
\hline
  $c_1(t, p)$ & $f_0(\yyp)$ & $\displaystyle 1/(2\xi^2 t)$ & 0\\
  $c_2(t, p)$ & $\displaystyle f_1(\yyp)$ & $\displaystyle 1/(2\xi^2 t)$ & 0\\
  $c_3(t, p)$ &   $\displaystyle \frac{6-5p}{6-4p}$ & $\displaystyle g_0(y^*)-\frac{\mu}{\xi^2 t}$ & $\displaystyle 2p-1$\\
  $c_4(t, p)$ & 0 & $\displaystyle g_1(y^*)-\frac{\mu}{\xi^2 t}-2$ & 0 \\
  $c_5(t, p)$ & 
$\displaystyle \frac{y_0^{\frac{p}{2}}\yyp^{\frac{3}{2}(1-p)}
\exp\left(\frac{k}{2} - \frac{y_0^{2(1-p)}}{2\xi^2t(1-p)^2} + \frac{f_1'(\yyp)^2}{2f_0''(\yyp)}\right)}{k^2\xi\sqrt{2\pi f_0''(\yyp)t}}$
 & $\displaystyle
\frac{\exp\left({\frac{k}{2}-\frac{\mu^2}{2\xi^2t}+\frac{\mu\log(y^*)}{\xi^2 t}}\right)}{4\sqrt{\pi}|k|^{-1}\xi^{-3}t^{-3/2}}$
& $\displaystyle \frac{2(p-1)
\E^{{-\frac{y_0^{2(p-1)}}{2\xi^2t(1-p)^2}}}\JJ^{2p}(k)}
{(2(1-p)^2\xi^2 t)^{\eta+1}\Gamma(\eta+1)}$\\
\hline
\\
%\label{eq:cgen}
 $h_1(\tau,p)$ & $\displaystyle \tau^{2(p-1)/(3-2p)}$ 
 & $\displaystyle \left(\log(\tau) + \log\log\left(\tau^{-1}\right)\right)^2$ & 0\\
 $h_2(\tau,p)$ & $\displaystyle  \tau^{(p-1)/(3-2p)}$
 & $\displaystyle \frac{(\log|\log (\tau)|)^2 }{|\log(\tau)|}$ & 0\\
\hline
\\
$\mathcal{R}(\tau, p)$ & $\displaystyle \mathcal{O}\left(\tau^{(1-p)/(3-2p)}\right)$
& $\displaystyle \mathcal{O}\left(\frac{1}{|\log(\tau)|}\right)$
& $\displaystyle \mathcal{O}(\tau^{p-1})$
\\
\hline
\end{tabular}
\bigskip
\captionof{table}{List of constants and functions} \label{tab:Table}
\end{minipage}

The following theorem (proved in Section~\ref{sec:ProofSection1CEV}) is the central result of this paper 
(although its equivalent below, in terms of implied volatility, is more informative for practical purposes):

\begin{theorem}\label{theorem:genfwdst}
The following expansion holds for all $k\in\RR\setminus\{0\}$ as $\tau$ tends to zero:
$$
\mathbb{E}\left(\E^{Z_{\tau}} - \E^{k}\right)^+
 = (1-\E^k)^+
+ 
\exp\Big(-c_1(t,p)h_1(\tau,p)+c_2(t,p)h_2(\tau,p)\Big)
\tau^{c_3(t,p)} |\log(\tau)|^{c_4(t,p)} c_5(t,p)
\left[1+\mathcal{R}(\tau, p)\right].
$$
\end{theorem}

\begin{remark}\ 
\begin{enumerate}[(i)]
\item Whenever $p\leq 1$, $c_1$ and $c_2$ are strictly positive;
the function~$c_5$ is always strictly positive;
when $p<1$, $c_3$ is strictly positive;
when $p=1$, the functions~$c_3$ and~$c_4$ can take positive and negative values;
\item Whenever $p\leq 1$, $h_2(\tau,p) \leq h_1(\tau,p)$ for $\tau$ small enough,
so that the leading order is provided by~$h_1$;
\item In the lognormal case $p=1$, $h_1(\tau,1)\sim (\log\tau)^2$ as $\tau$ tends to zero, so that the exponential decay of option prices is governed at leading order by $\exp(-c_1(t,1)(\log\tau)^2)$.
\end{enumerate}
\end{remark}

Using Theorem~\ref{theorem:genfwdst} and small-maturity asymptotics for the Black-Scholes model 
(see~\cite[Corollary 3.5]{FJL12} or~\cite{GL11}),
it is straightforward to translate option price asymptotics into asymptotics of the implied volatility:
\begin{theorem}\label{theorem:genfwdsm}
For any $k\in\RR\setminus\{0\}$, the small-maturity implied volatility smile behaves as follows:
\begin{align*}
\sigma_{\tau}^2(k)  \sim 
\left\{
\begin{array}{ll}
\displaystyle (1-\beta_p)\left(\frac{k^2\xi^2 t(1-p)}{2\tau}\right)^{\beta_p}, & \text{if }p <1,\\
\displaystyle \frac{k^2 \xi^2 t}{\tau \log(\tau)^2 }, & \text{if }p=1, \\
\displaystyle \frac{k^2}{2(2p-1) \tau |\log(\tau)|}, & \text{if }p>1.
\end{array}
\right.
\end{align*}
\end{theorem}
This theorem only presents the leading-order asymptotic behaviour of the implied volatility
as the maturity becomes small.
One could in principle (following~\cite{Caravenna} or~\cite{GL11, Guli, MijTankov}) derive higher-order terms,
but these additional computations would impact the clarity of this singular behaviour.
In the at-the-money $k=0$ case, the implied volatility converges to a constant:
\begin{lemma}\label{lem:explosiongenrem}
The at-the-money implied volatility $\sigma_{\tau}(0)$ converges to $\mathbb{E}(\sqrt{\Vv})$
as $\tau$ tends to zero.
\end{lemma}
The proof of the lemma follows steps analogous to~\cite[Lemma 4.3]{JR2013},
and we omit the details here.
It in fact does not depend on any particular form of distribution of~$\sqrt{\Vv}$,
as long as the expectation exists.
Note that, from Theorem~\ref{theorem:genfwdsm}, as~$p$ approaches $1$ from below, the rate of explosion approaches $\tau^{-1}$.
When~$p$ tends to~$1$ from above, the explosion rate is $1/(\tau|\log\tau|)$ instead.
So there is a "discontinuity" at $p=1$ and the actual rate of explosion is less than both these limits.
As an immediate consequence of Theorem~\ref{theorem:genfwdst} we have the following  corollary.
Define the following functions:
\begin{equation*}
h^*(\tau,p)  := 
\left\{
\begin{array}{ll}
\displaystyle \tau^{1-\beta_p}, & \text{if }p<1, \\
\displaystyle \left|\log(\tau)^{-1}\right|, & \text{if }p>1,\\
\displaystyle \log(\tau)^{-2}, & \text{if }p=1,
\end{array}
\right.
\qquad\text{and}\qquad
\Lambda^*_p(k) :=
\left\{
\begin{array}{ll}
\displaystyle c_1(t,p), & \text{if }p \leq 1, \\
\displaystyle 2p-1, & \text{if }p > 1,
\end{array}
\right.
\end{equation*}
where $c_1(t,p)$ is defined in Table~\ref{tab:Table}, and depends on~$k$ (through~$\overline{y}_p$).
\begin{corollary}\label{cor:LDP}
For any $p\in \RR$, the sequence~$\left(Z_\tau\right)_{\tau\geq0}$ satisfies 
a large deviations principle with speed $h^*(\tau,p)$ and rate function $\Lambda^*_p$ as~$\tau$ tends to zero.
Furthermore, the rate function is good only when $p<1$.
\end{corollary}

Recall that a real-valued sequence $(\mathfrak{Z}_n)_{n\geq 0}$ satisfies a large deviations principle (see~\cite{DZ93} for a precise introduction to the topic) with speed~$n$ and rate function~$\Lambda^*$ if,
for any Borel subset~$B\subset\RR$, 
the inequalities
$$
-\inf_{z\in B^o}\Lambda^*(z) \leq \liminf_{n\uparrow\infty}n^{-1}\log\PP(\mathfrak{Z}_n\in B)
\leq \limsup_{n\uparrow\infty}n^{-1}\log\PP(\mathfrak{Z}_n\in B)
\leq -\inf_{z\in \overline{B}}\Lambda^*(z) 
$$
hold, where $\overline{B}$ and $B^o$ denote the closure and interior of~$B$ in~$\RR$.
The rate function~$\Lambda^*:\RR\to\RR\cup\{+\infty\}$, by definition, 
is a lower semi-continuous, non-negative and not identically infinite, function
such that the level sets $\{x\in\RR: \Lambda^*(x)\leq \alpha\}$ are closed for all $\alpha\geq 0$.
It is said to be a good rate function when these level sets are compact (in~$\RR$).

\begin{proof}
The proof of Theorem~\ref{theorem:genfwdst} holds with only minor modifications for digital options,
which are equivalent to probabilities of the form 
$\mathbb{P}\left(Z_{\tau}\leq k\right)$ or $\mathbb{P}\left(Z_{\tau}\geq k\right)$.
For $p\in(-\infty,1]$,
one can then show that 
$$
\lim_{\tau\downarrow 0}h^*(\tau,p)\log \mathbb{P}\left(Z_{\tau}\leq k\right)
 = -\inf\left\{\Lambda^*_p(x): x\leq k\right\}.
$$
The infimum is null whenever $k>0$ and $p<1$, 
and $\Lambda^*_1(x)\equiv1/(2\xi^2t)$ is constant.
Consider now an open interval $(a,b)\subset\RR$.
Since $(a,b) = (-\infty, b)\setminus (-\infty,a]$, then by continuity and convexity of~$\Lambda^*_p$, 
we obtain
$$
\lim_{\tau\downarrow 0}h^*(\tau,p)\log \mathbb{P}\left(Z_{\tau} \in (a,b)\right)
 = -\inf_{x\in (a,b)}\Lambda_p^*(x).$$
The corollary then follows from the definition of the large deviations principle~\cite[Section 1.2]{DZ93}.
When $p\in(1,\infty)$, the only non-trivial choice of speed is $|(\log\tau)^{-1}|$, in which case
$\lim_{\tau\downarrow 0}|(\log\tau)^{-1}|\log \mathbb{P}\left(Z_{\tau}\leq k\right)
 = - (2p-1)$.
Clearly, the constant function is a rate function (the level sets, either the empty set or the real line,
being closed in~$\RR$), and the corollary follows.
\end{proof}

\begin{remark}
In the case $p=1/2$, as discussed in Section~\ref{sec:MGFSmallTime}, 
the moment generating function of~$Z$ is available in closed form.
However, the large deviations principle does not follow from the G\"artner-Ellis theorem,
since the pointwise rescaled limit of the mgf is degenerate (in the sense of~\eqref{eq:SmallTimeMGF12}).
\end{remark}

%%%%%%%%%%%%%%%%%%%%%%%%%%%%%%%%%%%%%%%%%%
%%%%%%%%%%%%%%%%%%%%%%%%%%%%%%%%%%%%%%%%%%
\subsubsection{Small-maturity at-the-money skew and convexity}

The goal of this section is to compute asymptotics for the at-the-money skew
and convexity of the smile as the maturity becomes small.
These quantities are useful to traders who actually observe them (or approximations thereof)
on real data.
We define the left and right derivatives by 
$\partial^{-}_{k}\sigma_{\tau}^2(0) := \lim_{k\uparrow 0}\partial_{k}\sigma_{\tau}^2(k)|_{k=0}$ and
$\partial^{+}_{k}\sigma_{\tau}^2(0) := \lim_{k\downarrow 0}\partial_{k}\sigma_{\tau}^2(k)|_{k=0}$, 
and similarly 
$\partial^{-}_{kk}\sigma_{\tau}^2(0) := \lim_{k\uparrow 0}\partial_{kk}\sigma_{\tau}^2(k)|_{k=0}$ and
$\partial^{+}_{kk}\sigma_{\tau}^2(0) := \lim_{k\downarrow 0}\partial_{kk}\sigma_{\tau}^2(k)|_{k=0}$.
The following lemma describes this short-maturity behaviour in the general case where $\Vv$ is any
random variable supported on $[0,\infty)$.
\begin{lemma}
Consider~\eqref{eq:Model} and assume that $\mathbb{E}(\Vv^{n/2})<\infty$ for $n=-1,1,3$,
and $\mm_t:=\PP(\Vv = 0) <1$.
As $\tau$ tends to zero,
\begin{equation*}
\begin{array}{rcllcl}
\partial^{\pm}_{k}\sigma_{\tau}^2(0)
 & \sim & \displaystyle 
\pm\frac{\mm_t \mathbb{E}(\sqrt{\Vv})\sqrt{\pi}}{\sqrt{2\tau}}, \\
  \partial^{-}_{kk}\sigma_{\tau}^2(0)
 & \sim & \partial^{+}_{kk}\sigma_{\tau}^2(0) \sim
\displaystyle \frac{\mathbb{E}(\sqrt{\Vv})}{\tau}\left(\mathbb{E}\left(\Vv^{-1/2}\right)
 - \mathbb{E}(\sqrt{\Vv})^{-1}\left(1-\frac{\mm_t^2\sqrt{\pi}}{8}\right)\right). \\
\end{array}
\end{equation*}
\end{lemma}
When $\mm_t>0$, the at-the-money left skew explodes to $-\infty$ and the at-the-money right skew explodes to $+\infty$.
Furthermore, the small-maturity at-the-money convexity tends to infinity.

\begin{proof}
We first focus on the at-the-money skew.
By definition the Call option price with log-moneyness~$k$ and maturity~$\tau$ reads
$C(k, \tau)=\BS(k, \sigma_{\tau}^2(k), \tau)$, and therefore
$$
\partial_k C(k, \tau) = \partial_k \BS(k, \sigma_{\tau}^2(k), \tau)  +
\partial_{k}\sigma_{\tau}^2(k) \partial_w \BS(k, \sigma_{\tau}^2(k), \tau).
$$
Also by~\eqref{eq:fstartintrep},
an immediate application of Fubini yields
\begin{equation}\label{eq:skewformulacev}
\partial_k C(k, \tau) = \int_{0}^{\infty}\partial_k\BS(k, y, \tau)\PP(\Vv \in \D y) + \mm_t \partial_k\left(1-\E^{k}\right)^{+},
\end{equation}
We first assume that $\mm_t=0$.
The at-the-money skew is then given by
\begin{equation}\label{eq:atmSkew0}
\partial_{k}\sigma_{\tau}^2(k)|_{k=0}
=\left(\partial_w \BS(k, \sigma_{\tau}^2(k), \tau)|_{k=0}\right)^{-1}
\left(
\int_{0}^{\infty}\partial_k\BS(k, y, \tau)|_{k=0}\PP(\Vv \in \D y) - \partial_k \BS(k, \sigma_{\tau}^2(k), \tau)|_{k=0}\right).
\end{equation}
Note now that, for any fixed~$y>0$ and any integer~$N$,
$$
\BS(0, y, \tau) = 1-\Nn\left(\frac{\sqrt{y\tau}}{2}\right)
 = \frac{1}{2} + \sum_{n=0}^{N}\alpha_{n} y^{n+1/2}\tau^{n+1/2} + \mathcal{O}\left(\tau^{N+3/2}\right),
$$
for some explicit sequence $(\alpha_{n})_{n\geq 0}$, as~$\tau$ tends to zero, 
and therefore,
$$
C(0, \tau) = \int_{0}^{\infty}\BS(0, y, \tau)\PP(\Vv\in\D y)
 = \frac{1}{2} + \sum_{n=0}^{N}\alpha_{n} \tau^{n+1/2}\EE\left(\Vv^{n+1/2}\right) + \mathcal{O}\left(\tau^{N+3/2}\right).
$$
This allows us to refine Lemma~\ref{lem:explosiongenrem}, 
so that, for any~$N$, there exists a sequence $(\sigma_n)_{n=0,\ldots,N}$
such that 
$$
\sigma_{\tau}(0) = \sum_{n=0}^{N}\sigma_n \tau^n + \mathcal{O}\left(\tau^{N+1}\right).
$$
We are now interested in the at-the-money skew.
Since 
$$
\partial_k\BS(k, \sigma_{\tau}^2(k), \tau)|_{k=0}
 = -\Nn\left(-\frac{\sigma_{\tau}(0)\sqrt{\tau}}{2}\right),
 $$
clearly admits the same (modulo signs) expansion as the Call price, it follows from~\eqref{eq:atmSkew0}
that the difference on the right will always be zero, and the lemma follows.
When $\mm_t>0$, we need to take left and right derivatives to account for the atomic term.
Since 
$\partial^{-}_{k}\left(1-\E^{k}\right)^{+}|_{k=0}=\partial^{-}_{kk}\left(1-\E^{k}\right)^{+}|_{k=0} = -1$ and 
$\partial^{+}_{k}\left(1-\E^{k}\right)^{+}|_{k=0}=\partial^{+}_{kk}\left(1-\E^{k}\right)^{+}|_{k=0}=0$,
the asymptotic skew stated in the lemma follows immediately.

The small-maturity convexity follows similar arguments, which we only outline.
Since
\begin{align*}
\partial_{kk} C(k, \tau) &= \partial_{kk} \BS(k, \sigma_{\tau}^2(k), \tau)  +
2\partial_{k}\sigma_{\tau}^2(k) \partial_{wk} \BS(k, \sigma_{\tau}^2(k), \tau)
+\left(\partial_{k}\sigma_{\tau}^2(k)\right)^2 \partial_{ww} \BS(k, \sigma_{\tau}^2(k), \tau) \\ \nonumber
&+\partial_{kk}\sigma_{\tau}^2(k) \partial_w \BS(k, \sigma_{\tau}^2(k), \tau),\\
\partial_{kk} C(k, \tau) & = \int_{0}^{\infty}\partial_{kk}\BS(k, y, \tau)\PP(\Vv \in \D y) + \mm_t \partial_{kk}\left(1-\E^{k}\right)^{+},
\end{align*}
then
\begin{align*}
\partial_{kk}\sigma_{\tau}^2(0)
  & = \left(\partial_w \BS(0, \sigma_{\tau}^2(0), \tau)\right)^{-1}
\left\{\partial_{kk} C(0, \tau)  - \partial_{kk} \BS(0, \sigma_{\tau}^2(0), \tau)\right\}\\
  & = \left(\partial_w \BS(0, \sigma_{\tau}^2(0), \tau)\right)^{-1}
\left\{\int_{0}^{\infty}\partial_{kk}\BS(k, y, \tau)\PP(\Vv \in \D y) - \partial_{kk} \BS(0, \sigma_{\tau}^2(0), \tau) + \mm_t \partial_{kk}\left(1-\E^{k}\right)^{+}\right\},
\end{align*}
and the lemma follows by straightforward computations similar to the skew case.
\end{proof}

%%%%%%%%%%%%%%%%%%%%%%%%%%%%%%%%%%%%%%%%%%%%%%%%
%%%%%%%%%%%%%%%%%%%%%%%%%%%%%%%%%%%%%%%%%%%%%%%%
\subsection{Large-time behaviour of option prices and implied volatility}\label{sec:LargeTime}

In this section we compute the large-time behaviour of option prices and implied volatility.
The proofs are given in Section~\ref{sec:largetimeproofcevvariance}.
It turns out that asymptotics are degenerate in the sense that option prices decay algebraically to their intrinsic values.
The structure of the asymptotic depends on the CEV parameter~$p$ and whether the origin is reflecting or absorbing: 
\begin{theorem}\label{thm:largetimecevatm}
Define the following quantity:
$$
\mathfrak{M}(\eta) := 
\frac{2^{3-6p-\eta}\Gamma\left(\frac{1}{2}-2p\right)}
{\sqrt{\pi}\Gamma(1+\eta)|1-p|^{2\eta+1} (\xi^2 t)^{\eta+1}}\exp\left(-\frac{y_0^{2(1-p)}}{2\xi^2 t (1-p)^2}\right),
$$
with~$\eta$ given in~\eqref{eq:constantsCEV}.
The following expansions hold for all $k\in\mathbb{R}$ as $\tau$ tends to infinity:
\begin{enumerate}[(i)]
\item if $p<3/4$ and the origin is absorbing then
$$
\mathbb{E}\left(\E^{Z_{\tau}} - \E^{k}\right)^+
 = 1 - \mm_t+\mm_t(1-\E^{k})^{+}
- 8\E^{k/2}y_0 \left(\frac{1}{2}-2p\right)\mathfrak{M}(-\eta)
\frac{1+\mathcal{O}\left(\tau^{-1}\right)}{\tau^{2-2p}};
$$
\item if $p<1/4$ and the origin is reflecting then
$$
\mathbb{E}\left(\E^{Z_{\tau}} - \E^{k}\right)^+
 = 1 - \E^{k/2}\mathfrak{M}(\eta)
\frac{1+\mathcal{O}\left(\tau^{-1}\right)}{\tau^{1-2p}}.
$$
\end{enumerate}
\end{theorem}
For other values of $p$, asymptotics are more difficult to derive and we leave this for future research.
The asymptotic behaviour of option prices is fundamentally different to Black-Scholes asymptotics (Lemma~\ref{lem:BSAsymplargetimecomplete}) and it is not clear that one can deduce asymptotics
for the implied volatility.
For example, the intrinsic values do not necessarily match as $\tau$ tends to infinity because of the mass at the origin.
The one exception is when the origin is reflecting, 
in which case the implied volatility tends to zero.
This result is a direct translation of Theorem~\ref{thm:largetimecevatm} into implied volatility asymptotics: \begin{theorem}\label{thm:IVLargeTime}
If $p<1/4$ and the origin is reflecting, the following pointwise limit holds for all $k\in\mathbb{R}$:
$$
\lim_{\tau\uparrow\infty}\frac{\tau}{\log\tau}\sigma_{\tau}^{2}(k) = 8(1-2p).
$$
\end{theorem}
\begin{proof}
One could prove the statement directly by computing the asymptotic behaviour of the Black-Scholes Call price
$\BS(k, y, \tau)$ as the maturity~$\tau$ tends to infinity (pointwise, for any $k>0$, $y\in\RR$),
see Lemma~\ref{lem:BSAsymplargetimecomplete},
and comparing it to the Call price expansion in Theorem~\ref{thm:largetimecevatm}(ii).
Instead, we choose to apply Tehranchi's result~\cite{Tehranchi}, which is the first fully model-independent 
study of the large-maturity behaviour of the implied volatility.
Assuming (a) that the underlying stock price~$\exp(Z)$ is a non-negative local martingale under~$\PP$, 
and (b) that $\exp(Z_t)$ converges almost surely to zero as time tends to infinity, 
Tehranchi~\cite[Theorem 3.1]{Tehranchi} proved that the expansion
\begin{equation}\label{eq:Tehranchi}
\sigma_{\tau}^2(k) = -8\log\EE\left(\E^{Z_\tau}\wedge \E^{k}\right)
 - 4\log\left\{-\log\EE\left(\E^{Z_\tau}\wedge \E^{k}\right)\right\} + 4k - 4\log\pi + \varepsilon(k,\tau)
\end{equation}
then holds uniformly on compact subsets of the real line as~$\tau$ tends to infinity,
where the function~$\varepsilon(\cdot)$ accounts for higher-order error terms.
It is clear here that the two assumptions above are satisfied in our model.
Note that (b) is equivalent to Call prices converging to one as the maturity tends to infinity
(see~\cite[Lemma 3.3]{RogersTehranchi} for instance).
Using the almost sure equality $(\E^{Z_\tau}-\E^{k})_+ = \E^{Z_\tau} - \E^{Z_\tau}\wedge \E^{k}$, 
and the fact that $(\E^{Z_\tau})_{\tau\geq 0}$ is a true martingale, 
it is then straightforward to show that Theorem~\ref{thm:IVLargeTime} follows from
Theorem~\ref{thm:largetimecevatm}(ii) and Tehranchi's expansion~\eqref{eq:Tehranchi}.
\end{proof}
\begin{remark}
In fact, the proof of Theorem~\ref{thm:IVLargeTime} actually provides higher-order terms, but we omit them here for brevity.
The case of Theorem~\ref{thm:IVLargeTime}(i) shows that the Call option price converges to
\begin{equation*}
1 - \mm_t + \mm_t(1-\E^{k})^{+} =
\left\{
\begin{array}{ll}
1 - \mm_t, & \text{if } k\geq 0,\\
1 - \mm_t\E^{k}, & \text{if } k< 0,
\end{array}
\right.
\end{equation*}
as the maturity tends to infinity.
This clearly is never equal to zero since $\mm_t \in (0,1)$, so that Tehranchi's Assumption (b) (in the proof of Theorem~\ref{thm:IVLargeTime}) fails.
\end{remark}
Although we provided here the large-time behaviour of the implied volatility, 
it is not our intention to use this model for options with large expiries.
Our intention (as mentioned in Section~\ref{sec:introcevvar}) is to use it as a building block for more advanced models (such as stochastic volatility models where the initial variance is sampled from a continuous distribution) so that we are able to better match steep small-maturity observed smiles.
In these types of more sophisticated models, the large-time behaviour is governed more from the chosen stochastic volatility model rather than the choice of distribution for the initial variance (see~\cite{JR12, JR14} for examples),
especially if the variance process possesses some ergodic properties.
This also suggests to use this class of models to introduce two different time scales: 
one to match the small-time smile (the distribution for the initial variance) 
and one to match the medium- to large-time smile (the chosen stochastic volatility model).
We leave this particular point for further (on-going) research, 
and direct the interested reader to Section~\ref{sec:forwardcev}, 
where more intuition about the use of this framework for forward-start options.

%%%%%%%%%%%%%%%%%%%%%%%%%%%%%%%%%%%%%%%%%%%%%%%%
%%%%%%%%%%%%%%%%%%%%%%%%%%%%%%%%%%%%%%%%%%%%%%%%
\subsection{Numerics}~\label{sec:numericsgenfwdsmile}
We provide here two types of numerical examples.
In Section~\ref{sec:Num_BSCEVSurf}, we show how randomising the Black-Scholes model 
according to~\eqref{eq:Model}
distorts the standard flat Black-Scholes implied volatility surface,
and generates a realistic-looking one.
In Section~\ref{sec:Num_Asymptotic}, we compare numerically the asymptotic results 
for the implied volatility smile to the true smile generated from\eqref{eq:Model}.

\subsubsection{Black-Scholes-CEV surface}\label{sec:Num_BSCEVSurf}
We consider here the following values:
\begin{equation}\label{eq:paramBSCEV}
t = 1, \quad
\xi = 20\%, \quad
p = 0.5, \quad
y_0 = 10\%, \quad
S_0 = 1.
\end{equation}
In Figure~\ref{fig:BSCEVSurface}, we plot the implied volatility surface generated by~\eqref{eq:Model}
according to the values given in~\eqref{eq:paramBSCEV}.
First, note that, contrary to the standard Black-Scholes model, the surface is not flat.
Second, and more importantly, the smile becomes steeper and more pronounced as the maturity becomes small.
This is a widely recognised fact on Equity markets, and seems to validate the approach followed in this paper.
Note that, following Section~\ref{sec:LeeWingsMGF}, one can taylor the parameters of the CEV component
in order to match any desired (arbitrage-free) slope for the wings of the smile.
\begin{figure}[h]
\centering
\includegraphics[scale=0.5]{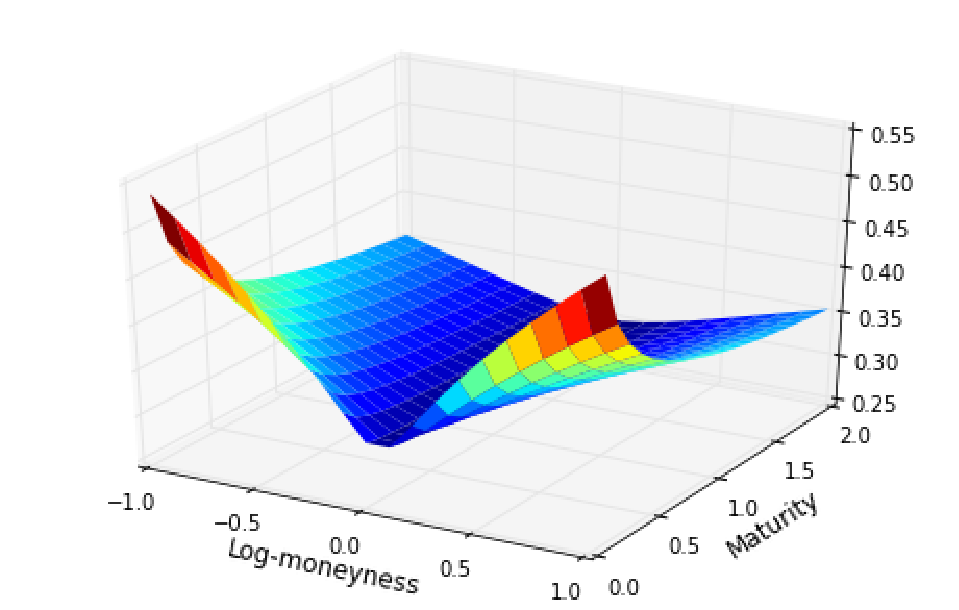}
\caption{
BS-CEV implied volatility surface, with parameters being given in~\eqref{eq:paramBSCEV}.}
\label{fig:BSCEVSurface}
\end{figure}

\subsubsection{Asymptotics}\label{sec:Num_Asymptotic}
We calculate option prices using the representation~\eqref{eq:fstartintrep} 
and a global adaptive Gauss-Kronrod quadrature scheme. 
We then compute the smile with a simple root-finding algorithm.
In Figure~\ref{fig:RandomBSSurface}(a), (b) and (c) we plot the smile for different maturities and values for the CEV power~$p$.
The model parameters are 
$y_0=0.07$, $\xi=0.2 y_0^{1/2-p}$ and $t=1/2$.
Note here that we set $\xi$ to be a different value for each $p$.
This is done so that the models are comparable: $\xi$ is then given in the same units and the quadratic variation of the CEV variance dynamics are approximately matched for different values of $p$.
The graphs highlight the steepness of the smiles as the maturity gets smaller and the role of $p$ in the shape of the small-maturity smile.
Note (as mentioned in previous sections) that the random variance acts as a shock to the small-maturity volatility surface and then flattens out.
The shape of the shock depends on the CEV power, $p$. 
Out-of-the money volatilities (for $K\notin[0.9,1.1]$) explode at a quicker rate as~$p$ increases 
(this can be seen from Theorem~\ref{theorem:genfwdsm}).
The volatility for strikes close to the money $K\in[0.9,1.1]$ appears to be less explosive as one increases $p$,
which might be explained from the strike dependence of the coefficients of the asymptotic in Theorem~\ref{theorem:genfwdsm}.
In order to compare our asymptotic to the true smile we use Theorem~\ref{theorem:genfwdst} to extend Theorem~\ref{theorem:genfwdst} to higher order. 
For the case $p<1,k\neq0$ we find that
$\sigma_{\tau}^2(k)\sim a_0(k)\tau^{-\beta_p}+a_1(k)\tau^{-p\beta_p}$ as $\tau$ tends to zero with
$$
a_0(k):=(1-\beta_p)\left(\frac{k^2\xi^2 t(1-p)}{2}\right)^{\beta_p},\qquad
a_1(k)=\frac{2f_1(y_p)a_0(k)^2}{k^2},
$$
and $\beta_p, y_p$ and $f_1$ defined in~\eqref{eq:betapalphasmmat}-\eqref{eq:f0g0smmat}.
At first order we see a close match with the true smile in Figure~\ref{fig:RandomBSSurface}(d).

\begin{figure}
\centering
\mbox{\subfigure[p=0.2.]{\includegraphics[scale=0.7]{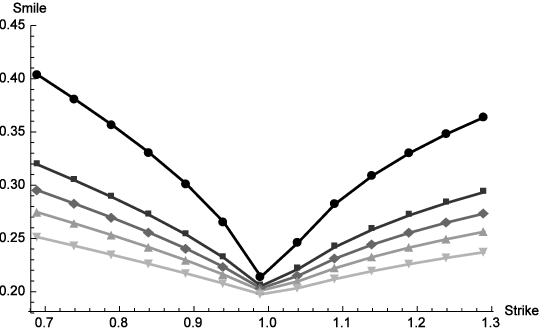}}\quad
\subfigure[p=1]{\includegraphics[scale=0.7]{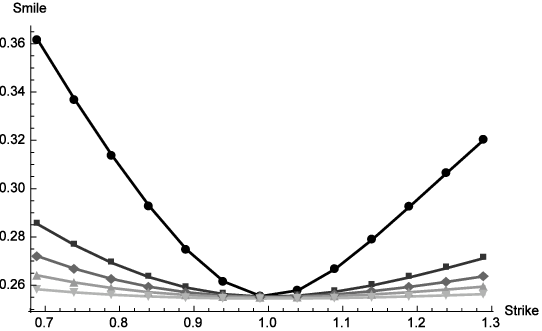}}}
\mbox{\subfigure[p=1.3]{\includegraphics[scale=0.7]{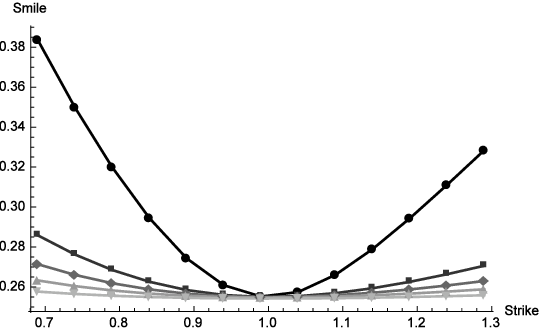}} \quad
\subfigure[Actual vs asymptotic]{\includegraphics[scale=0.7]{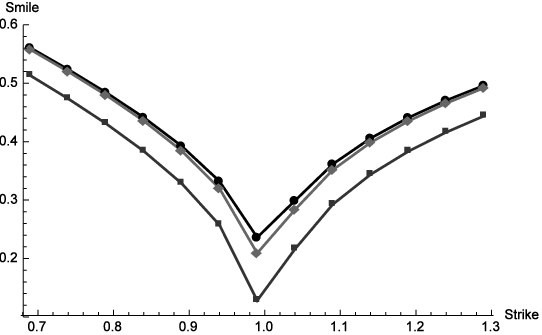}}}
\caption{
In (a), (b), (c) we plot 
$K\mapsto\sigma_{\tau}(\log K)$ for maturities of 1/12 (circles),1/2 (squares),1 (diamonds),2 (triangles) and 5 (backwards triangles) for increasing values of the CEV power~$p$.
In (d) we plot the actual small maturity smile for $p=0.2$  and $\tau=1/100$ (circles) and the zeroth (squares) and first order (diamonds) smile using Theorem~\ref{theorem:genfwdsm}.
Parameters of the model are given in the text.}
\label{fig:RandomBSSurface}
\end{figure}

%%%%%%%%%%%%%%%%%%%%%%%%%%%%%%%%%%%%%%%%%%%%%%%%
%%%%%%%%%%%%%%%%%%%%%%%%%%%%%%%%%%%%%%%%%%%%%%%%
\subsection{Application to forward smile asymptotics}\label{sec:forwardcev}
We now show how our model~\eqref{eq:Model} and the asymptotics derived above for the implied volatility
can be directly translated into asymptotics of the forward implied volatility in stochastic volatility models.
For a given martingale process~$\E^{X}$, a forward-start option with reset date~$t$, maturity~$t+\tau$
and strike~$\E^{k}$ is worth, at inception, $\EE(\exp(X_{t+\tau}-X_t) - \E^{k})_+$.
In the Black-Scholes model, the stationarity property of the increments imply that this option is simply
equal to a standard Call option on~$\E^{X}$ (started at $X_0=0$) with strike~$\E^{k}$ and maturity~$\tau$;
therefore, one can define the forward implied volatility~$\sigma_{t,\tau}(k)$, similarly to the standard implied volatility (see~\cite{JR12} for more details).
Suppose now that the log stock price process $X$ satisfies the following SDE:
\begin{equation}\label{eq:SVGeneral}
\begin{array}{rll}
\D X_s & = \displaystyle -\frac{1}{2}Y_s\D s+ \sqrt{Y_s}\D W_s, \quad & X_0=0,\\
\D Y_s & = \xi_{s} Y_s^p\D B_s, \quad & Y_0=y_0>0,\\
\D\left\langle W,B\right\rangle_s & = \rho \D s,
\end{array}
\end{equation}
with $p\in\mathbb{R}$, $|\rho|<1$ and $W, B$ are two standard Brownian motions.
Fix the forward-start date~$t>0$ and set
\begin{align*}
\xi_u : =
\left\{
\begin{array}{ll}
\displaystyle \xi, & \text{if }0\leq u \leq t, \\
\displaystyle \bar{\xi}, &\text{if } u>t,
\end{array}
\right.
\end{align*}
where $\xi>0$ and $\bar{\xi}\geq 0$.
This includes the Heston ($p=1/2$) and 3/2 ($p=3/2$) models with zero mean reversion 
as well as the SABR model ($p=1$).
Let $X_{\tau}^{(t)}:=X_{t+\tau}-X_{t}$ denote the forward price process and let $\CEV(t,\xi,p)$ be the distribution such that $\textrm{Law}(Y_t) = \textrm{Law}(\Vv) = \CEV(t,\xi,p)$.
Then the following lemma holds:
\begin{lemma}
In the model~\eqref{eq:SVGeneral} the forward price process $X_{\cdot}^{(t)}$ solves the following system of SDEs:
\begin{equation}
\begin{array}{rll}
\D X_\tau^{(t)} & = \displaystyle -\frac{1}{2} Y_\tau^{(t)}\D \tau+ \sqrt{ Y_\tau^{(t)}}\D W_\tau, \quad &  X_0^{(t)}=0,\\
\D Y_\tau^{(t)} & = \bar{\xi} \left(Y_{\tau}^{(t)}\right)^p\D B_\tau, \quad & Y_0^{(t)}\sim \CEV(t,\xi,p),\end{array}
\end{equation}
where $Y_0^{(t)}$ is independent to the Brownian motions $(W_{\tau})_{\tau\geq0}$ and $(B_{\tau})_{\tau\geq0}$.
\end{lemma}
This lemma makes it clear that forward-start options in stochastic volatility models are European options on a stock price with similar dynamics to~\eqref{eq:SVGeneral}, but with initial variance sampled from the variance distribution at the forward-start date.
When $\bar{\xi}=0$, then $X_{\cdot}^{(t)}=Z$ and forward smile asymptotics follow immediately:
\begin{corollary}\label{cor:fwdsmilegene}
If $\bar{\xi}=0$, Theorem~\ref{theorem:genfwdst},
Theorem~\ref{theorem:genfwdsm} and Lemma~\ref{lem:explosiongenrem} hold with $Z=X_{\cdot}^{(t)}$ and $\sigma_{\tau}=\sigma_{t,\tau}$.
\end{corollary}

\begin{remark}\label{rem:5Points}\
\begin{enumerate}[(i)]
\item Corollary~\ref{cor:fwdsmilegene} explicitly links the shape and fatness of the right tail of the variance distribution at the forward-start date and the asymptotic form and explosion rate of the small-maturity forward smile. 
Take for example $p>1$: the density of the variance in the right wing is dominated by the polynomial~$y^{-2p}$ and the exponential dependence on~$y$ is irrelevant. 
So the smaller $p$ in this case, the fatter the right tail and hence the larger the coefficient of the expansion.
This also explains the algebraic (not exponential) dependence for forward-start option prices.
\item The asymptotics in the $p>1$ case are extreme and the algebraic dependence on~$\tau$ 
is similar to small-maturity exponential L\'evy models.
This extreme nature is related to the fatness of the right tail of the variance distribution: 
for example, the $3/2$ model ($p=3/2$) allows for the occurrence of extreme paths with periods of very high instantaneous volatility (see~\cite[Figure 3 ]{D12}).
\item The asymptotics in Theorems~\ref{theorem:genfwdst} and~\ref{theorem:genfwdsm} remain the same (at this order) regardless of whether the variance process is absorbing or reflecting at zero when $p\in(-\infty,1/2)$.
Intuitively this is because absorption or reflection primarily influences the left tail whereas small-maturity forward smile asymptotics are influenced by the shape of the right tail of the variance distribution.
\item When $p=1/2$ in Corollary~\ref{cor:fwdsmilegene}, 
the asymptotics are the same as in~\cite[Theorem 4.1]{JR2013} for the Heston model.
This seems to indicate that the key quantity determining the small-maturity forward smile explosion rate is the variance distribution at the forward-start date, and not the actual dynamics of the stock price.
\item 
Practitioners~\cite{Ba,B02} have stated that the Heston model ($p=1/2$) produces small-maturity forward smiles that are too convex and "U-shaped" and inconsistent with observations, 
but that SABR-or lognormal-based models ($p=1$) produce less convex or "U-shaped" small-maturity forward smiles. 
Our results provide theoretical insight into this effect.
We observed in Section~\ref{sec:numericsgenfwdsmile} and Figure ~\ref{fig:RandomBSSurface} that the explosion effect was more stable for strikes close to the money as one increased~$p$.
The strike dependence of the asymptotic implied volatility in Theorem~\ref{theorem:genfwdsm} is given by $K\mapsto \sqrt{|\log K|}$ for $p=1/2$ and $K\mapsto |\log K|$ for $p=1$.
It is clear from the figures that the forward implied volatility is more U-shaped for $p\geq 1$.
\end{enumerate}
\end{remark}

%%%%%%%%%%%%%%%%%%%%%%%%%%%%%%%%%%%%%%%%%%%%%%%%
%%%%%%%%%%%%%%%%%%%%%%%%%%%%%%%%%%%%%%%%%%%%%%%%
\section{Proofs}\label{sec:ProofsCEVgeneral}

\subsection{Proof of Theorem~\ref{theorem:genfwdst}}\label{sec:ProofSection1CEV}

Let $C(k, \tau):=\mathbb{E}(\E^{Z_{\tau}} - \E^k)^+$.
This function clearly depends on the parameter~$t$, but we omit this dependence in the notations.
The tower property implies
\begin{equation}\label{eq:fstartintrep}
C(k, \tau) = \int_{0}^{\infty}\BS(k, y, \tau)\zeta_p(y)\D y + \mm_t\left(1-\E^{k}\right)^{+},
\end{equation}
where $\BS$ is defined in~\eqref{eq:BSvariance}, 
$\zeta_p$ is density of $\Vv$ given in~\eqref{eq:CEVdensity} 
and $\mm_t$ is the mass at the origin~\eqref{eq:mass0}.
Our goal is to understand the asymptotics of this integral as $\tau$ tends to zero.
We break the proof of Theorem~\ref{theorem:genfwdst} into three parts: 
in Section~\ref{sec:smmatp1} we prove the case $p>1$, in Section~\ref{subsection:p01proof} we prove the case $p\in(-\infty,1)$ and in Section~\ref{sec:smmatsecp1} we prove the case $p=1$.
We only prove the result for $k>0$, the arguments being completely analogous when~$k<0$.
The key insight is that one has to re-scale the variance in terms of the maturity $\tau$ before asymptotics can be computed.
The nature of the re-scaling depends critically on the CEV power $p$ and fundamentally different asymptotics result in each case.
Note that for $k>0$, $\left(1-\E^{k}\right)^{+}=0$, so that the atomic term in~\eqref{eq:fstartintrep} 
is irrelevant for the analysis.
When $k<0$ the arguments are analogous by Put-Call parity.

%%%%%%%%%%%%%%%%%%%%%%%%%%%%%%%%%%
%%%%%%%%%%%%%%%%%%%%%%%%%%%%%%%%%%
\subsubsection{Case: $p>1$}\label{sec:smmatp1}
In Lemma~\ref{lem:CEVBoundspgreater1} we prove a bound on the CEV density.
This is sufficient to allow us to prove asymptotics for option prices in Lemma~\ref{lem:fstpgreaterone} 
after rescaling the variance by~$\tau$.
This rescaling is critical because it is the only one making $\BS(k, y/\tau, \tau)$ independent of $\tau$.
Let 
$$\chi(\tau,p):=\frac{\tau^{2p}}
{|1-p|\xi^2t\Gamma(1+|\eta|)\Big(2(1-p)^2\xi^2 t\Big)^{|\eta|}}
\exp\left({-\frac{y_0^{2(1-p)}}{2\xi^2t(1-p)^2}}\right),
$$ 
and we have the following lemma:
\begin{lemma}\label{lem:CEVBoundspgreater1}
The following bounds hold for the CEV density for all $y,\tau>0$ when $p>1$:
\begin{align*}
\zeta_p\left(\frac{y}{\tau}\right)
 & \geq \frac{\chi(\tau,p)}{y^{2p}}
\left\{1-\frac{1}{2\xi^2t(1-p)^2}\left(\frac{\tau}{y}\right)^{2p-2}\right\}, \\ 
\zeta_p\left(\frac{y}{\tau}\right)
 & \leq \frac{\chi(\tau,p)}{y^{2p}}
\left\{1+\exp\left({\frac{y_0^{2-2p}}{2(p-1)^2t \xi^2}}\right)
\left[\frac{1}{2\xi^2t(1-p)^2}\left(\frac{\tau}{y}\right)^{2p-2}
+\frac{1}{\xi^2t(1-p)^2}\left(\frac{\tau}{y y_0}\right)^{p-1}\right]\right\}.
\end{align*}
\end{lemma}
\begin{proof}
From~\cite[Equation (6.25)]{L72} we know that for $x>0$ and $\nu>-1/2$, 
the modified Bessel function satisfies
\begin{equation}\label{eq:CEVdensitybounds}
\frac{1}{\Gamma(\nu+1)}\left(\frac{x}{2}\right)^\nu \leq I_{\nu}(x)\leq \frac{\E^{x}}{\Gamma(\nu+1)}\left(\frac{x}{2}\right)^\nu,
\end{equation}
so that the expression for the CEV density in~\eqref{eq:CEVdensity} 
implies that for $p>1$,
$$
\frac{\chi(\tau,p)}{y^{2p}}\exp\left({-\frac{1}{2\xi^2t(1-p)^2}\left(\frac{\tau}{y}\right)^{2p-2}}\right)
\leq \zeta_p\left(\frac{y}{\tau}\right)\leq
\frac{\chi(\tau,p)}{y^{2p}}\E^{m(y,\tau)},
$$
where
$$
m(y,\tau):=
-\frac{1}{2\xi^2t(1-p)^2}\left(\frac{\tau}{y}\right)^{2p-2}
+\frac{1}{\xi^2t(1-p)^2}\left(\frac{\tau}{y y_0}\right)^{p-1}.
$$
For fixed $\tau>0$, note that $m(\cdot,\tau):\mathbb{R}_{+}\mapsto\mathbb{R}_{+}$ 
takes a maximum positive value at $y=y_0\tau$ with $m(y_0\tau,\tau) = y_0^{2-2p}/(2(p-1)^2t \xi^2)$.
When $m(\cdot)>0$ Taylor's Theorem with remainder yields
$\E^{m(y,\tau)}=1+\E^{\gamma}m(y,\tau)$ for some $\gamma\in (0,m(y,\tau))$, 
and hence $\E^{m(y,\tau)}\leq 1+\E^{m(y_0\tau,\tau)}m(y,\tau)$.
If $m(\cdot)<0$ then $\E^{m(y,\tau)}\leq 1+|m(y,\tau)| \leq 1+\E^{m(y_0\tau,\tau)}|m(y,\tau)|$.
The result for the upper bound then follows by the triangle inequality for $|m(y,\tau)|$.
The lower bound simply follows from the inequality $1-x\leq \E^{-x}$, valid for $x>0$, and 
$$
1-\frac{1}{2\xi^2t(1-p)^2}\left(\frac{\tau}{y}\right)^{2p-2}\leq \exp\left({-\frac{1}{2\xi^2t(1-p)^2}\left(\frac{\tau}{y}\right)^{2p-2}}\right).
$$
\end{proof}

\begin{lemma}\label{lem:fstpgreaterone}
When $p>1$, Theorem~\ref{theorem:genfwdst} holds.
\end{lemma}
\begin{proof}
The substitution $y\to y/\tau$ and ~\eqref{eq:fstartintrep} imply that
the option price reads 
$C(k, \tau)
 = \int_0^{\infty}\BS(k, y, \tau)\zeta_{p}(y)\D y
 = \tau^{-1}\int_0^{\infty}\BS(k, y/\tau, \tau)\zeta_{p}(y/\tau)\D y$.
Using Lemma~\ref{lem:CEVBoundspgreater1} and Definition~\eqref{eq:Ip},
we obtain the following bounds:
\begin{align*}
\frac{\chi(\tau,p)}{\tau}\left[\JJ^{2p}(k)-\frac{\tau^{2p-2}}{2\xi^2t(1-p)^2}\JJ^{4p-2}(k)\right]
\leq C(k, \tau), \\
\frac{\chi(\tau,p)}{\tau}
\left[\JJ^{2p}(k) + \exp\left({\frac{y_0^{2-2p}}{2(p-1)^2t \xi^2}}\right)
\left(\frac{\tau^{2p-2}}{2\xi^2t(1-p)^2}\JJ^{4p-2}(k)
 + \frac{\tau^{p-1}}{\xi^2t(1-p)^2 y_0^{p-1}}\JJ^{3p-1}(k)\right)\right] 
 \geq C(k, \tau).
\end{align*}
Hence for $\tau<1$:
$$
\left|\frac{C(k, \tau)\tau}{\chi(\tau,p)\JJ^{2p}(k)}-1\right|
\leq \exp\left({\frac{y_0^{2-2p}}{2(p-1)^2t \xi^2}}\right)
\left(\frac{\JJ^{4p-2}(k)}{2\xi^2t(1-p)^2 \JJ^{2p}(k)}
 + \frac{\JJ^{3p-1}(k)}{\xi^2t(1-p)^2 y_0^{p-1} \JJ^{2p}(k)}\right)
\tau^{p-1},
$$ 
which proves the lemma since $\JJ^{2p}(k)$ is strictly positive, finite and independent of $\tau$.
\end{proof}

%%%%%%%%%%%%%%%%%%%%%%%%%%%%%%%%%%%%%%%%%%%
%%%%%%%%%%%%%%%%%%%%%%%%%%%%%%%%%%%%%%%%%%%
\subsubsection{Case: $p<1$}\label{subsection:p01proof}

We use the representation in~\eqref{eq:fstartintrep} and break the domain of the integral up 
into a compact part and an infinite (tail) one.
We prove in Lemma~\ref{lem:Tailestimatep01} that the tail integral is exponentially sub-dominant 
(compared to the compact part) and derive asymptotics for the integral in Lemma~\ref{lem:fstartoptp01}. 
This allows us to apply the Laplace method to the integral.
We start with the following bound for the modified Bessel function of the first kind and then prove a tail estimate in Lemma~\ref{lem:Tailestimatep01}.

\begin{lemma}\label{lem:besselboundSeg}
The following bound holds for all $x>0$ and $\nu>-3/2$:
$$
\II_{\nu}(x)<\frac{\nu+2}{\Gamma(\nu+2)}\left(\frac{x}{2}\right)^{\nu}\E^{2x}.
$$
\end{lemma}
\begin{proof}
Let $x>0$.
From~\cite[Theorem 7, page 522]{S11}, the inequality
$\II_{\nu}(x)<\II_{\nu+1}(x)^2 / \II_{\nu+2}(x)$ holds whenever $\nu\geq-2$, and hence 
combining it with~\eqref{eq:CEVdensitybounds} (valid only for $\nu>-1/2$), we can write
$$
\II_{\nu}(x)<\frac{\Gamma(\nu+3)}{\Gamma(\nu+2)^2}\left(\frac{x}{2}\right)^{\nu}\E^{2x},
$$
when $\nu>-3/2$.
The lemma then follows from the trivial identity $\Gamma(\nu+3) = (\nu+2)\Gamma(\nu+2)$.
\end{proof}

\begin{lemma}\label{lem:Tailestimatep01}
Let $L>1$ and $p<1$.
Then the following tail estimate holds as $\tau$ tends to zero:
$$
\int_{L}^{\infty}\BS\left(k, \frac{y}{\tau^{\beta_p}}, \tau\right)\zeta_p\left(\frac{y}{\tau^{\beta_p}}\right)\D y
 = \mathcal{O}\left(\exp\left({-\frac{1}{4\xi^2t(1-p)}
 \left[\frac{L^{1-p}}{\tau^{(1-\beta_p)/2}} - y_0^{1-p}\right]^2}\right)\right).
$$
\end{lemma}
\begin{proof}
Lemma~\ref{lem:besselboundSeg} and the density in~\eqref{eq:CEVdensity} imply
$$
\zeta_p\left(\frac{y}{\tau^{\beta_p}}\right)
\leq
\frac{b_0}
{\tau^{-2p\beta_p }}
y^{-2p}
\exp\left(-\frac{1}{2\xi^2t(1-p)^2}\left\{\frac{y^{1-p}}{\tau^{\beta_p(1-p)}} - y_0^{1-p}\right\}^2
+\frac{(yy_0)^{1-p}}{\tau^{\beta_p(1-p)}\xi^2t(1-p)^2}\right),
$$
where the constant $b_0$ is given by 
$$
\frac{(\eta+2)}
{|1-p|\xi^2t\Gamma(\eta+2)\Big(2(1-p)^2\xi^2 t\Big)^{\eta}},
\qquad
\text{resp.}
\qquad
\frac{(|\eta|+2)}
{|1-p|\xi^2t\Gamma(|\eta|+2)\Big(2(1-p)^2\xi^2 t\Big)^{|\eta|}},
$$
if the origin is reflecting (resp. absorbing) when $p<1/2$;
the exact value of $b_0$ is however irrelevant for the analysis.
Set now $L>1$.
Using this upper bound and the no-arbitrage inequality $\BS(\cdot)\leq 1$, we find
\begin{align*}
 & \int_{L}^{\infty}\BS\left(k, \frac{y}{\tau^{\beta_p}}, \tau\right)
\zeta_p\left(\frac{y}{\tau^{\beta_p}}\right)\D y
\leq 
\int_{L}^{\infty}\zeta_p\left(\frac{y}{\tau^{\beta_p}}\right)\D y\\
 & \displaystyle \leq 
\frac{b_0}
{\tau^{-2p\beta_p }}
\int_{L}^{\infty}y^{-2p}
\exp\left(-\frac{1}{2\xi^2t(1-p)^2}\left\{\frac{y^{1-p}}{\tau^{\beta_p(1-p)}} - y_0^{1-p}\right\}^2
+\frac{(yy_0)^{1-p}}{\tau^{\beta_p(1-p)}\xi^2t(1-p)^2}\right)\D y \\
 & \displaystyle \leq 
\frac{b_0}
{\tau^{-2p\beta_p }}
\int_{L}^{\infty}y^{1-2p}
\exp\left(-\frac{1}{2\xi^2t(1-p)^2}\left\{\frac{y^{1-p}}{\tau^{\beta_p(1-p)}} - y_0^{1-p}\right\}^2
+\frac{(yy_0)^{1-p}}{\tau^{\beta_p(1-p)}\xi^2t(1-p)^2}\right)\D y,
\end{align*}
where the last line follows since $y^{1-2p}>y^{-2p}$.
Setting $q=\left(y^{1-p}/\tau^{\beta_p(1-p)} - y_0^{1-p}\right)/(\xi\sqrt{t}(1-p))$ yields
\begin{align}\label{eq:boundsreflecting}
&\int_{L}^{\infty}y^{1-2p}
\exp\left(-\frac{\left(\frac{y^{1-p}}{\tau^{\beta_p(1-p)}} - y_0^{1-p}\right)^2}{2\xi^2t(1-p)^2}
+\frac{(yy_0)^{1-p}}{\tau^{\beta_p(1-p)}\xi^2t(1-p)^2}\right)\D y \nonumber\\
 &= \frac{\xi\sqrt{t} (1-p)}{\tau^{2\beta_p(p-1)}}
 \left[\xi\sqrt{t}(1-p)\int_{L_\tau}^{\infty}q\exp\left[-\frac{q^2}{2}+\frac{y_0^{1-p}q}{\xi\sqrt{t}(1-p)}\right]\D q 
  + y_0^{1-p}\int_{L_\tau}^{\infty}\exp\left[-\frac{q^2}{2}+\frac{y_0^{1-p}q}{\xi\sqrt{t}(1-p)}\right]\D q\right],
\end{align}
with $L_\tau:=\left(L^{1-p}/\tau^{\beta_p(1-p)}-y_0^{1-p}\right)/(\xi\sqrt{t}(1-p)) >0$
for small enough $\tau$ since $L>1$ and $p\in(-\infty,1)$.
Set now (we always choose the positive root)
$$
\tau^*:=
\left(\frac{L^{1-p}}{5y_0^{1-p}}\right)^{\left(\beta_{p}(1-p)\right)^{-1}},
$$
so that, for $\tau<\tau^*$ we have $L_{\tau}>4y_0^{1-p}/(\xi\sqrt{t}(1-p))$ and hence for $q>L_{\tau}$:
$$
\frac{y_0^{1-p}q}{\xi\sqrt{t}(1-p}\leq \frac{q^2}{4}.
$$
In particular, for the integrals in~\eqref{eq:boundsreflecting} we have the following bounds for $\tau<\tau^*$:
\begin{align*}
&\int_{L_\tau}^{\infty}q\exp\left(-\frac{q^2}{2}+\frac{y_0^{1-p}q}{\xi\sqrt{t}(1-p)}\right)\D q
\leq
\int_{L_\tau}^{\infty}q\exp\left(-\frac{q^2}{4}\right)\D q
 = 2\exp\left(-\frac{L_\tau^2}{4}\right), \\
&\int_{L_\tau}^{\infty}\exp\left(-\frac{q^2}{2}+\frac{y_0^{1-p}q}{\xi\sqrt{t}(1-p)}\right)\D q
\leq
\int_{L_\tau}^{\infty}\exp\left(-\frac{q^2}{4}\right)\D q
\leq \frac{4}{L_{\tau}}\exp\left(-\frac{L_\tau^2}{4}\right),
\end{align*}
where the last inequality follows from the upper bound
for the complementary normal distribution function in~\cite[Section 14.8]{W03}.
The lemma then follows from noting that $1-\beta_p=2\beta_p(1-p)$.
\end{proof}

\begin{lemma}\label{lem:fstartoptp01}
When $p<1$, Theorem~\ref{theorem:genfwdst} holds.
\end{lemma}
\begin{proof}
Let $\ttau:=\tau^{\beta_p}$, with $\beta_p$ defined in~\eqref{eq:betapalphasmmat}.
Applying the substitution $y\to y/\ttau$ to~\eqref{eq:fstartintrep} yields
\begin{align*}
C(k, \tau)
 & = \int_0^{\infty}\BS(k, y, \tau)\zeta_p(y)\D y
 = \frac{1}{\ttau}\int_0^{\infty}\BS\left(k, \frac{y}{\ttau}, \tau\right)\zeta_p\left(\frac{y}{\ttau}\right)\D y\\
 & = \frac{1}{\ttau}\int_0^{L}
\BS\left(k, \frac{y}{\ttau},\tau\right)\zeta_p\left(\frac{y}{\ttau}\right)\D y
 + \frac{1}{\ttau}\int_L^{\infty}\BS\left(k, \frac{y}{\ttau}, \tau\right)\zeta_p\left(\frac{y}{\ttau}\right)\D y,
\end{align*}
for some $L>0$ to be chosen later.
We start with the first integral.
Using the asymptotics for the modified Bessel function of the first kind~\cite[Section 9.7.1]{Abra} 
as $\tau$ tends to zero, 
we obtain
$$
\zeta_p\left(\frac{y}{\ttau}\right)
 = \frac{\tau^{3p\beta_p/2}y_0^{p/2}\E^{-\frac{y_0^{2(1-p)}}{2\xi^2t(1-p)^2}}}{\xi y^{3p/2}\sqrt{2\pi t}}	
\exp\left(-\frac{1}{\tau^{2\beta_p(1-p)}}\frac{y^{2(1-p)}}{2\xi^2t(1-p)^2}
+\frac{1}{\tau^{\beta_p(1-p)}}\frac{(yy_0)^{(1-p)}}{\xi^2t(1-p)^2}\right)
\left[1+\mathcal{O}\left(\tau^{(1-p)\beta_p}\right)\right].
$$
Note that this expansion does not depend on the sign of~$\eta$ and so the same asymptotics hold 
regardless of whether the origin is reflecting or absorbing.
In the Black-Scholes model, Call option prices satisfy
(Lemma~\ref{lem:BSsmalltimecev}):
\begin{equation*}%\label{eq:SmallTimeBS}
\BS\left(k, \frac{y}{\ttau},\tau\right)
=
\frac{y^{3/2}}{k^2\sqrt{2\pi}}\left(\frac{\tau}{\ttau}\right)^{3/2}
\exp\left(-\frac{k^2}{2 y}\frac{\ttau}{\tau}+\frac{k}{2}\right)
\left(1 + \mathcal{O}\left(\frac{\tau}{\ttau}\right)\right),
\end{equation*}
as $\tau$ tends to zero.
Using the identity $1-\beta_p=2\beta_p(1-p)$ we then compute 
\begin{multline*}
\frac{1}{\tau^{\beta_p}}\int_{0}^{L}\BS\left(k, \frac{y}{\tau^{\beta_p}}, \tau\right) \zeta_p\left(\frac{y}{\tau^{\beta_p}}\right) \D y \\
=
\frac{\tau^{\beta_{p}(4-3p)/2}y_0^{p/2}\E^{-\frac{y_0^{2(1-p)}}{2\xi^2t(1-p)^2}+\frac{k}{2}}}{2\pi k^2\xi\sqrt{t}}
\int_{0}^{L}y^{\frac{3}{2}(1-p)} \E^{-\frac{f_0(y)}{\tau^{1-\beta_p}}
 + \frac{f_1(y)}{\tau^{(1-\beta_p)/2}}} \D y
\left[1+\mathcal{O}\left(\tau^{(1-\beta_p)/2}\right)\right],
\end{multline*}
where $f_0,f_1$ are defined in~\eqref{eq:f0g0smmat}.
Solving the equation $f'_0(y)=0$ gives $y=\yyp$ with $\overline{y}_p$ defined in~\eqref{eq:betapalphasmmat} 
and we always choose the positive root and set $L>\overline{y}_p$.
Let 
$I(\tau) : = 
\int_{0}^{L}y^{\frac{3}{2}(1-p)}
\exp\left(-\frac{f_0(y)}{\tau^{1-\beta_p}}+\frac{f_1(y)}{\tau^{(1-\beta_p)/2}}\right) \D y$. 
Then for some $\varepsilon>0$ small enough, as $\tau$ tends to zero, the asymptotic equivalences
\begin{align*}
I(\tau)
&\sim 
\exp\left(-\frac{f_0(\overline{y}_p)}{\tau^{1-\beta_p}}+\frac{f_1(\overline{y}_p)}{\tau^{(1-\beta_p)/2}}+\frac{f_1'(\overline{y}_p)^2}{2f_0''(\overline{y}_p)}\right)
\overline{y}_p^{\frac{3}{2}(1-p)}
\int_{\overline{y}_p-\varepsilon}^{\overline{y}_p+\varepsilon}\exp\left({-\frac{1}{2}\left[\frac{\sqrt{f''_0(\overline{y}_p)}(y-\overline{y}_p)}{\tau^{(1-\beta_p)/2}}-\frac{f_1'(\overline{y}_p)}{\sqrt{f''_0(\overline{y}_p)}}\right]^2} \right)
\D y  \\
&\sim
\exp\left(-\frac{f_0(\overline{y}_p)}{\tau^{1-\beta_p}}+\frac{f_1(\overline{y}_p)}{\tau^{(1-\beta_p)/2}}+\frac{f_1'(\overline{y}_p)^2}{2f_0''(\overline{y}_p)}\right)
\overline{y}_p^{\frac{3}{2}(1-p)}
\int_{-\infty}^{\infty}\exp\left({-\frac{1}{2}\left[\frac{\sqrt{f''_0(\overline{y}_p)}(y-\overline{y}_p)}{\tau^{(1-\beta_p)/2}}-\frac{f_1'(\overline{y}_p)}{\sqrt{f''_0(\overline{y}_p)}}\right]^2} \right)
\D y \\
& =
\exp\left(-\frac{f_0(\overline{y}_p)}{\tau^{1-\beta_p}}+\frac{f_1(\overline{y}_p)}{\tau^{(1-\beta_p)/2}}+\frac{f_1'(\overline{y}_p)^2}{2f_0''(\overline{y}_p)}\right)
\tau^{(1-\beta_p)/2}\overline{y}_p^{\frac{3}{2}(1-p)}\sqrt{\frac{2\pi}{f_0''(\overline{y}_p)}}.
\end{align*}
hold.
It follows that  as $\tau$ tends to zero:
$$
\frac{1}{\tau^{\beta_p}}\int_{0}^{L}\BS\left(k, \frac{y}{\tau^{\beta_p}}, \tau\right)
\zeta_p\left(\frac{y}{\beta_p}\right) \D y
=
\exp\left(-\frac{c_1(t,p)}{\tau^{1-\beta_p}}+\frac{c_2(t,p)}{\tau^{(1-\beta_p)/2}}\right)
c_5(t,p)\tau^{c_3(t,p)}
\left[1+\mathcal{O}\left(\tau^{\frac{1-\beta_p}{2}}\right)\right].
$$
From Lemma~\ref{lem:Tailestimatep01} we know that 
$$
\frac{1}{\tau^{\beta_p}}\int_{L}^{\infty}
\BS\left(k, \frac{y}{\tau^{\beta_p}}, \tau\right)\zeta_p(y/\beta_p) \D y
 = \mathcal{O}\left(\exp\left({-\frac{1}{2\xi^2t(1-p)}\left(\frac{L^{1-p}}{\tau^{(1-\beta_p)/2}}-y_0^{1-p}\right)^2}\right)\right).
$$
Choosing $L>\max\left(1,\left(2\xi^2 t(1-p)f_0(\yyp)\right)^{1/(2-2p)},\yyp\right)$ 
makes this tail term exponentially subdominant to 
$\tau^{-\beta_p}\int_{0}^{L}\BS(k, y/\tau^{\beta_p}, \tau) \zeta_p(y/\beta_p) \D y$,
which completes the proof of the lemma.
\end{proof}

%%%%%%%%%%%%%%%%%%%%%%%%%%%%%%%%%%%%%%
%%%%%%%%%%%%%%%%%%%%%%%%%%%%%%%%%%%%%%
\subsubsection{Case: $p=1$}\label{sec:smmatsecp1}

In the lognormal case $p=1$, the random variable~$\log(\Vv)$ is Gaussian
with mean $\mu$ (defined in~\eqref{eq:constantsCEV}) and variance $\xi^2 t$.
The proof is similar to Section~\ref{subsection:p01proof}, 
but we need to re-scale the variance by $\tau|\log(\tau)|$.
We prove a tail estimate in Lemma~\ref{lem:Tailestimatepequals1} and derive asymptotics 
for option prices in Lemma~\ref{lemma:fwdstartoptp1}.

\begin{lemma}\label{lem:Tailestimatepequals1}
The following tail estimate holds for $p=1$ and $L>0$ as $\tau$ tends to zero ($\mu$ defined in~\eqref{eq:constantsCEV}):
$$
\int_{L}^{\infty}\BS\left(k, \frac{y}{\tau|\log(\tau)|}, \tau\right)\zeta_1\left(\frac{y}{\tau|\log(\tau)|}\right)\D y
 = \mathcal{O}\left(\exp\left\{-\frac{1}{2\xi^2 t}\left[\log\left(\frac{L}{\tau|\log(\tau)|}\right)-\mu\right]^2\right\}\right).
$$
\end{lemma}
\begin{proof}
By no-arbitrage arguments, the Call price is always bounded above by one, so that
$$
\int_{L}^{\infty}\BS\left(k, \frac{y}{\tau|\log(\tau)|}, \tau\right)\zeta_1\left(\frac{y}{\tau|\log(\tau)|}\right)\D y
 \leq \int_{L}^{\infty}\zeta_1\left(\frac{y}{\tau|\log(\tau)|}\right)\D y.
$$
With the substitution $q = \frac{1}{\xi\sqrt{t}}[\log(y/(\tau|\log(\tau)|))-\mu]$,
the lemma follows from the bound for the complementary Gaussian distribution function~\cite[Section 14.8]{W03}.
\end{proof}

\begin{lemma}\label{lemma:fwdstartoptp1}
Let $p=1$.
The following expansion holds for option prices as $\tau$ tends to zero:
$$
C(k, \tau)
=
c_5(t,1)\exp\Big(-c_1(t,1)h_1(\tau,p)+c_2(t,1)h_2(\tau,p)\Big)
\tau^{c_3(t,1)} |\log(\tau)|^{c_4(t,1)}
\left(1+\mathcal{O}\left(\frac{1}{|\log(\tau)|}\right)\right),
$$
with the functions $c_1,c_2,...,c_5$, $h_1$ and $h_2$ given in Table~\ref{tab:Table}.
\end{lemma}

\begin{proof}
Let $\ttau := \tau|\log(\tau)|$.
With the substitution $y\to y/\ttau$  and using~\eqref{eq:fstartintrep}, 
the option price is given by
\begin{align*}
C(k, \tau) & = 
 \int_0^{\infty}\BS(k, y, \tau)\zeta_1(y)\D y
 = \frac{1}{\ttau}
\int_0^{\infty}\BS\left(k, \frac{y}{\ttau}, \tau\right)
\zeta_1\left(\frac{y}{\ttau}\right)\D y\\
 & = \frac{1}{\ttau}\left\{
\int_0^{L}\BS\left(k, \frac{y}{\ttau}, \tau\right)
\zeta_1\left(\frac{y}{\ttau}\right)\D y
 + \int_L^{\infty}\BS\left(k, \frac{y}{\ttau}, \tau\right)\zeta_1\left(\frac{y}{\ttau}\right)\D y\right\}
=:\underbar{C}(k, \tau) + \overline{C}(k, \tau),
\end{align*}
for some $L>0$.
Consider the first term.
Using Lemma~\ref{lem:BSsmalltimecev} with $\ttau = \tau|\log(\tau)|$, we have, as $\tau$ tends to zero,
$$
\BS\left(k, \frac{y}{\tau|\log(\tau)|}, \tau\right)
=
\exp\left(-{\frac{k^2 |\log(\tau)|}{2 y}+\frac{k}{2}}\right)
\frac{y^{3/2} }{k^2 |\log(\tau)|^{3/2}\sqrt{2\pi}}
\left[1+\mathcal{O}\left(\frac{1}{|\log(\tau)|}\right)\right].
$$
Therefore
\begin{align*}
\underbar{C}(k, \tau)
&=
\frac{\E^{k/2}\left(1+\mathcal{O}\left(\frac{1}{|\log(\tau)|}\right)\right)}{|\log(\tau)|^{3/2}\xi k^2 2\pi \sqrt{t}} 
\int_{0}^{L}\exp\left(-\frac{k^2 |\log(\tau)|}{2y } - \frac{\left(\log\left(\frac{y}{\tau|\log(\tau)|}\right)-\mu\right)^2}{2\xi^2 t}\right) y^{1/2}
\D y \\
&=
\exp\left(
\frac{k}{2}
 - \frac{(\log(\tau) + \log|\log(\tau)|)^2+\mu^2}{2\xi^2 t}
 - \frac{\mu(\log(\tau) + \log|\log(\tau)|)}{\xi^2t}\right)
\frac{I_{1}(\tau)\left[1+\mathcal{O}\left(\frac{1}{|\log(\tau)|}\right)\right]}
{\xi k^2 2\pi \sqrt{t}|\log(\tau)|^{3/2}},
\end{align*}
where
$I_{1}(\tau):=\int_{0}^{L} g_2(y) \exp\left(-g_0(y)|\log\tau|+g_1(y)\log|\log(\tau)|\right)\D y$ with
$g_0$ and $g_1$ defined in~\eqref{eq:f0g0smmat} and
$$
g_2(y):=\sqrt{y}\exp\left(\frac{\mu\log(y)}{\xi^2 t}\right).
$$
The dominant contribution from the integrand is the $|\log(\tau)|$ term;
the minimum of~$g_0$ is attained at~$y^*$ given in~\eqref{eq:betapalphasmmat}, and 
$g''_0(y^*)=4/(\xi^6 t^3 k^4)>0$.
Set
$$
I_0(\tau):=
\int_{-\infty}^{\infty} 
\exp\left(-\frac{1}{2}\left((y-y^*)\sqrt{|\log(\tau)|g_0''(y^*)} - \frac{g'(y^*)\log|\log(\tau)|}
{\sqrt{|\log(\tau)|g_0''(y^*)}}\right)^2\right)
\D y  
  = \sqrt{\frac{2 \pi}{g_0''(y^*)|\log(\tau)|}}.
$$
Then for some $\varepsilon>0$ as $\tau$ tends to zero, the asymptotic equivalences
with $L>y^*$,
\begin{align*}
I_{1}(\tau) &\sim
\int_{y^*-\epsilon}^{y^*+\epsilon} g_2(y) 
\exp\Big\{-g_0(y)|\log(\tau)| +g_1(y)\log|\log(\tau)|\Big\}
\D y \\
&\sim
g_2(y^*)  \E^{-g_0(y^*)|\log(\tau)|+g_1(y^*)\log|\log(\tau)|} 
\int_{y^*-\epsilon}^{y^*+\epsilon} 
\E^{-\frac{1}{2}g_0''(y^*)(y-y^*)^2|\log(\tau)|  + g_1'(y^*)(y-y^*)\log|\log(\tau)|}
\D y\\
&\sim
g_2(y^*) \exp\left({-g_0(y^*)|\log(\tau)|+g_1(y^*)\log|\log(\tau)|
+ \frac{(g_1'(y^*)\log|\log(\tau)| )^2}{2g_0''(y^*)|\log(\tau)|}  }\right)
I_0(\tau) \\
& =
g_2(y^*)  \exp\left({-g_0(y^*)|\log(\tau)|+g_1(y^*)\log|\log(\tau)|
+ \frac{( g_1'(y^*)\log|\log(\tau)|)^2}{2g_0''(y^*)|\log(\tau)|} }\right)
\sqrt{\frac{2 \pi}{g_0''(y^*)|\log(\tau)|}}.
\end{align*}
hold.
Therefore as $\tau$ tends to zero:
$$
\underbar{C}(k, \tau)=
c_5(t,1) \exp\Big(-c_1(t,1)h_1(\tau,1)+c_2(t,1)h_2(\tau,1)\Big)
\tau^{c_3(t,1)} |\log(\tau)|^{c_4(t,1)}
\left[1+\mathcal{O}\left(\frac{1}{|\log(\tau)|}\right)\right],
$$
with the functions $c_1,c_2,...,c_5$, $h_1$ and $h_2$ given in Table~\ref{tab:Table}.
For ease of computation we note that
$$
c_5(t,1)=\frac{\sqrt{y^*}\exp\left(\frac{k}{2}-\frac{\mu^2}{2\xi^2t}+\frac{\mu\log(y^*)}{\xi^2 t}\right)}{k^2\xi\sqrt{2\pi t}\sqrt{g_0''(y^*)}}
=\frac{|k|\xi^3t^{3/2}\exp\left(\frac{k}{2}-\frac{\mu^2}{2\xi^2t}+\frac{\mu\log(y^*)}{\xi^2 t}\right)}{4\sqrt{\pi}}.
$$
Now by Lemma~\ref{lem:Tailestimatepequals1},
\begin{align*}
\overline{C}(k, \tau)
 &= \frac{1}{\tau|\log(\tau)|}\int_{L}^{\infty}\BS\left(k, \frac{y}{\tau|\log(\tau)|}, \tau\right)
\zeta_1\left(\frac{y}{\tau|\log(\tau)|}\right)\D y \\
&\qquad \qquad \qquad \qquad =\frac{1}{\tau|\log(\tau)|}\mathcal{O}\left(\exp\left\{-\frac{1}{2\xi^2 t}\left[\log\left(\frac{L}{\tau|\log(\tau)|}\right)-\mu\right]^2\right\}\right).
\end{align*}
Since for some $B>0$ we have that
$$
\exp\left({-\frac{1}{2\xi^2 t}\left[\log\left(\frac{L}{\tau|\log(\tau)|}\right)-\mu\right]^2}\right)
\leq
B
\left(\tau|\log(\tau)|\right)^{\frac{1}{\xi^2t}\left(\log(L)-\mu\right)}
\exp\left({-\frac{1}{2\xi^2t}h_1(\tau,1)}\right),
$$
choosing $L$ such that $\log(L)>\mu$ yields
$$
\mathcal{O}\left(\exp\left\{-\frac{1}{2\xi^2 t}\left[\log\left(\frac{L}{\tau|\log(\tau)|}\right)-\mu\right]^2\right\}\right)
=
\mathcal{O}\left(\exp\left(-\frac{1}{2\xi^2t}h_1(\tau,1)\right)\right).
$$
Hence $\overline{C}(k, \tau)$ is then exponentially subdominant to the compact part since 
$$
\exp\Big(c_1(t,1)h_1(\tau,1)-c_2(t,1)h_2(\tau,1)\Big)
 \mathcal{O}\left(\exp\left\{-\frac{1}{2\xi^2 t}\left[\log\left(\frac{L}{\tau|\log(\tau)|}\right)-\mu\right]^2\right\}\right) 
= \mathcal{O}\left(\E^{-c_2(t,1)h_2(\tau,1)}\right),
$$
and the result follows.
\end{proof}

%%%%%%%%%%%%%%%%%%%%%%%%%%%%%%%%%%%%%%%%
\subsection{Proof of Theorem~\ref{thm:largetimecevatm}}\label{sec:largetimeproofcevvariance}
Lemma~\ref{lem:BSAsymplargetimecomplete} and~\eqref{eq:fstartintrep} 
yield the following asymptotics as $\tau$ tends to infinity:
\begin{equation}\label{eq:largetimecevariancecall}
C(k, \tau)
 = 1-\mm_t+\mm_t(1-\E^{k})^{+}+
\tau^{-1/2}\E^{k/2}\mathfrak{L}(\tau)(1+\mathcal{O}(\tau^{-1})),
\end{equation}
where
$\mathfrak{L}(\tau) := \int_0^{\infty}q(z)\E^{-\tau z}\D z$,
and we set $q(z)\equiv-8\zeta_p(8z)/\sqrt{\pi z}.$
As $z$ tends to zero recall the following asymptotics for the modified Bessel function of the first kind of order $\eta$~\cite[Section 9.6.10]{Abra}:
$$
\II_{\eta}(z)=\frac{1}{\Gamma(\eta+1)}\left(\frac{z}{2}\right)^{\eta}\left(1+\mathcal{O}\left(z^2\right)\right).
$$
Using this asymptotic and~ the definition of the density in~\eqref{eq:CEVdensity}
we obtain the following asymptotics for the density as $y$ tends to zero when $p<1$ 
and absorption at the origin when $p<1/2$:
\begin{equation}\label{eq:zetaabszero}
\zeta_p(y)=\frac{y_0y^{1-2p}}{|1-p|\xi^2t \Gamma(|\eta|+1)\left(2(1-p)^2\xi^2 t\right)^{|\eta|}}
\exp\left(-\frac{y_0^{2(1-p)}}{2\xi^2 t (1-p)^2}\right)
\left(1+\mathcal{O}\left(y^{2(1-p)}\right)\right).
\end{equation}
Analogous arguments yield that when $p<1/2$ and the origin is reflecting, 
then, as $y$ tends to zero,
\begin{equation}\label{eq:zetarefzero}
\zeta_p(y)=\frac{y^{-2p}}{|1-p|\xi^2t \Gamma(\eta+1)\left(2(1-p)^2\xi^2t\right)^{\eta}}
\exp\left(-\frac{y_0^{2(1-p)}}{2\xi^2 t (1-p)^2}\right)
\left(1+\mathcal{O}\left(y^{2(1-p)}\right)\right).
\end{equation}
In order to apply Watson's lemma~\cite[Part 2, Chapter 2]{Miller} to~$\mathfrak{L}$,
it suffices to require that 
$q(z)=\mathcal{O}(\E^{c z})$ for some $c>0$ as $z$ tends to infinity.
This clearly holds here since $\lim_{z\uparrow\infty}\zeta_p(z)=0$.
We also require that $q(z)=a_0z^{l}(1+\mathcal{O}(z^n))$ as~$z$ tends to zero for some $l>-1$, $n>0$.
When $p\geq1$, it can be shown that~$\zeta_p$ is exponentially small, 
and a different method needs to be used.
When $p<1$ and the density is as in~\eqref{eq:zetaabszero} then $l=1-2p-\frac{1}{2}$ and so we require $p<\frac{3}{4}$.
Analogously, when $p<1/2$ and the density is~\eqref{eq:zetarefzero} then $l=-2p-\frac{1}{2}$ and we require $p<\frac{1}{4}$.
An application of Watson's Lemma in conjunction with~\eqref{eq:largetimecevariancecall} yields Theorem~\ref{thm:largetimecevatm}.

%%%%%%%%%%%%%%%%%%%%%%%%%%%%%%%%%%%%%%%%%%%%
\begin{appendix}
\section{Black-Scholes asymptotics}
This appendix gathers some useful expansions for the Black-Scholes Call price 
function~$\BS$ defined in~\eqref{eq:BSvariance}.

\begin{lemma}\label{lem:BSsmalltimecev}
Let $k,y>0$ and $\ttau:(0,\infty)\to (0,\infty)$ be a continuous function such that 
$\displaystyle \lim_{\tau\downarrow 0}\frac{\tau}{\ttau(\tau)}=0$. 
Then
$$
\BS\left(k, \frac{y}{\ttau(\tau)},\tau\right)
=
\frac{y^{3/2}}{k^2\sqrt{2\pi}}\left(\frac{\tau}{\ttau(\tau)}\right)^{3/2}
\exp\left(-\frac{k^2}{2 y}\frac{\ttau(\tau)}{\tau}+\frac{k}{2}\right)
\left\{1+\mathcal{O}\left(\frac{\tau}{\ttau(\tau)}\right)\right\},
\quad\text{as }\tau \text{ tends to zero}.
$$
\end{lemma}
\begin{proof}
Let $k,y>0$ and set $\tau^*(\tau) \equiv \tau/{\ttau}(\tau)$.
By assumption, $\tau^*(\tau)$ tends to zero, 
and~\eqref{eq:BSvariance} implies
$$
\BS\left(k, \frac{y}{\ttau},\tau\right)= 
\BS\left(k, y,\tau^*(\tau)\right)
=\Nn(d_+^*(\tau)) - \E^k\Nn(d_-^*(\tau)),$$
where we set $d_\pm^*(\tau) := -k/(\sqrt{y\tau^*(\tau)})\pm \frac{1}{2}\sqrt{y\tau^*(\tau)}$. 
Note that $d_\pm^*$ tends to $-\infty$ as $\tau$ tends to zero.
The asymptotic expansion
$1-\Nn(z)=(2\pi)^{-1/2}\E^{-z^2/2}\left(z^{-1}-z^{-3}+\mathcal{O}(z^{-5})\right)$,
valid for large~$z$ (\cite[page 932]{Abra}), 
yields
\begin{align*}
\BS\left(k, \frac{y}{\ttau(\tau)},\tau\right)
 &  = \Nn\left(d_+^*(\tau)\right)-\E^k\Nn\left(d_-^*(\tau)\right)
 =1-\Nn\left(-d_+^*(\tau)\right)-\E^k(1-\Nn\left(-d_-^*(\tau)\right)) \\
&= \frac{1}{\sqrt{2\pi}}\exp\left(-\frac{1}{2}d_{+}^*(\tau)^2/2\right)
\left\{\frac{1}{d_{-}^*(\tau)} -\frac{1}{d_{+}^*(\tau)}
 + \frac{1}{d_{+}^*(\tau)^3} - \frac{1}{d_{-}^*(\tau)^3}
  + \mathcal{O}\left(\frac{1}{d_{+}^*(\tau)^5}\right)\right\},
\end{align*}
as $\tau$ tends to zero, where we used the identity
$\frac{1}{2}d_{-}^*(\tau)^2 - k = \frac{1}{2}d_{+}^*(\tau)^2$.
The lemma then follows from the following expansions as $\tau$ tends to zero:
\begin{align*}
\exp\left(-\frac{1}{2}d_{+}^*(\tau)^2\right)
 & = \exp\left(-\frac{k^2}{2y\tau^*} + \frac{k}{2}\right)\left(1+\mathcal{O}(\tau^*(\tau))\right),\\
\frac{1}{d_{-}^*(\tau)} -\frac{1}{d_{+}^*(\tau)}
 + \frac{1}{d_{+}^*(\tau)^3} - \frac{1}{d_{-}^*(\tau)^3}
  & = \frac{y^{3/2} \tau^*(\tau)^{3/2}}{k^2}\left(1+\mathcal{O}(\tau^*(\tau))\right).
\end{align*}
\end{proof}

\begin{lemma}\label{lem:BSAsymplargetimecomplete}
Let $y>0$ and $k\in\RR$. Then, as $\tau$ tends to infinity,
$$
\BS(k, y, \tau)
 = 1 - \frac{4}{\sqrt{2\pi\tau y}}\exp\left(-\frac{y\tau}{8} + \frac{k}{2}\right)\left(1+\mathcal{O}(\tau^{-1})\right).
$$
\end{lemma}
\begin{proof}
Let $y>0$.
Then
$\BS(k, y, \tau)=\Nn\left(d_+^*(\tau)\right)-\E^k\Nn\left(d_-^*(\tau)\right)$,
where
$d_\pm^*(\tau) := -k/(\sqrt{y\tau})\pm \frac{1}{2}\sqrt{y\tau}$,
Hence $d_{\pm}^*$ tends to $\pm\infty$ as $\tau$ tends to infinity.
Similarly to the proof of the previous lemma,
\begin{align*}
\BS(k, y, \tau)
 & = \Nn\left(d_+^*(\tau)\right)-\E^k\left(1-\Nn\left(-d_-^*(\tau)\right)\right) \\
 & = 1- \frac{1}{\sqrt{2\pi}}\exp\left(-\frac{1}{2}d_{+}^*(\tau)^2\right)
\left\{\frac{1}{d_{+}^*(\tau)} -\frac{1}{d_{-}^*(\tau)}
 + \frac{1}{d_{-}^*(\tau)^3} - \frac{1}{d_{+}^*(\tau)^3}
  + \mathcal{O}\left(\frac{1}{d_{+}^*(\tau)^5}\right)\right\},
\end{align*}
as $\tau$ tends to infinity, where we used the identity
$\frac{1}{2}d_{-}^*(\tau)^2 - k = \frac{1}{2}d_{+}^*(\tau)^2$.
The lemma then follows from the following expansions as $\tau$ tends to infinity:
\begin{align*}
\exp\left(-\frac{1}{2}d_{+}^*(\tau)^2\right)
 & = \exp\left(-\frac{y\tau}{8} + \frac{k}{2}\right)\left(1+\mathcal{O}(\tau^{-1})\right),\\
\frac{1}{d_{+}^*(\tau)} -\frac{1}{d_{-}^*(\tau)}
 + \frac{1}{d_{-}^*(\tau)^3} - \frac{1}{d_{+}^*(\tau)^3}
  & = \frac{4}{\sqrt{2\pi\tau y}}\left(1+\mathcal{O}(\tau^{-1})\right).
\end{align*}
\end{proof}

\end{appendix}

%%%%%%%%%%%%%%%%%%%%%%%%%%%%%%%%%%%%%%%%%%
%%%%%%%%%%%%%%%%%%%%%%%%%%%%%%%%%%%%%%%%%%

\end{document}